\documentclass[prd,
showpacs,
onecolumn,floatfix,superscriptaddress]{revtex4-1}
\usepackage{times,amsmath,amsfonts,amssymb,latexsym}
\usepackage{graphicx,epsf}

\newcommand{\bra}[1]{\langle #1 |}
\newcommand{\ket}[1]{| #1 \rangle}

\newcommand{\be}{\begin{equation}}
\newcommand{\ee}{\end{equation}}
\newcommand{\ba}{\begin{eqnarray}}
\newcommand{\ea}{\end{eqnarray}}

\newcommand{\ignore}[1]{}
\def\CC{{\rm\kern.24em \vrule width.04em height1.46ex depth-.07ex
    \kern-.30em C}}
\def\P{{\rm I\kern-.25em P}}
\def\RR{{\rm
         \vrule width.04em height1.58ex depth-.0ex
         \kern-.04em R}}

\def\bbbc{{\mathchoice {\setbox0=\hbox{$\displaystyle\rm C$}\hbox{\hbox
to0pt{\kern0.4\wd0\vrule height0.9\ht0\hss}\box0}}
{\setbox0=\hbox{$\textstyle\rm C$}\hbox{\hbox
to0pt{\kern0.4\wd0\vrule height0.9\ht0\hss}\box0}}
{\setbox0=\hbox{$\scriptstyle\rm C$}\hbox{\hbox
to0pt{\kern0.4\wd0\vrule height0.9\ht0\hss}\box0}}
{\setbox0=\hbox{$\scriptscriptstyle\rm C$}\hbox{\hbox
to0pt{\kern0.4\wd0\vrule height0.9\ht0\hss}\box0}}}}
\def\bbbz{{\mathchoice {\hbox{$\sf\textstyle Z\kern-0.4em Z$}}
{\hbox{$\sf\textstyle Z\kern-0.4em Z$}}
{\hbox{$\sf\scriptstyle Z\kern-0.3em Z$}}
{\hbox{$\sf\scriptscriptstyle Z\kern-0.2em Z$}}}}

\newtheorem{theorem}{Theorem}[section]

\newtheorem{corollary}[theorem]{Corollary}

\newenvironment{proof}[1][Proof]{\begin{trivlist}
\item[\hskip \labelsep {\bfseries #1}]}{\end{trivlist}}

\begin{document}

\title{A quantum Bose-Hubbard model with evolving graph as toy model for emergent spacetime}

\author{Alioscia Hamma}
\email{ahamma@perimeterinstitute.ca}
\author{Fotini Markopoulou}
\affiliation{Perimeter Institute for Theoretical Physics,
31 Caroline St. N, N2L 2Y5, Waterloo Ontario, Canada}
\email{fmarkopoulou@perimeterinstitute.ca}
\author{Seth Lloyd}
\affiliation{Massachusetts Institute of Technology,
77 Massachusetts Avenue, Cambridge Massachusetts, 02139 USA}
\email{slloyd@mit.edu}
\author{Francesco Caravelli}
\affiliation{Perimeter Institute for Theoretical Physics,
31 Caroline St. N, N2L 2Y5, Waterloo Ontario, Canada}
\affiliation{University of Waterloo, 200 University Ave W, Waterloo ON, N2L 3G1, Canada }
\email{fcaravelli@perimeterinstitute.ca}
\author{Simone Severini}
\affiliation{Department of Physics and Astronomy, University College London, WC1E 6BT London, United Kingdom}
\email{simoseve@gmail.com}
\author{Klas Markstr\"om }
\affiliation{Department of Mathematics and Mathematical Statistics,
Ume� universitet,
S-901 87 Umea, Sweden}

\begin{abstract}
We present a toy model for interacting matter and geometry that explores quantum dynamics in a spin system as a precursor to a quantum theory of gravity. The model has no \emph{a priori} geometric properties; instead, locality is inferred from the more fundamental notion of interaction between the matter degrees of freedom. The interaction terms are themselves quantum degrees of freedom so that the structure of interactions and hence the resulting local and causal structures are dynamical. The system is a Hubbard model where the graph of the interactions is a set of quantum evolving variables. We show entanglement between spatial and matter degrees of freedom. We study numerically the quantum system and analyze its entanglement dynamics. We analyze the asymptotic behavior of the classical model. Finally, we discuss analogues of trapped surfaces and gravitational attraction in this simple model.
\end{abstract}

 \maketitle
\section{Introduction}

The quest for a quantum theory of gravity involves searching for a microscopic quantum theory whose low energy limit is the known physics of general relativity, dynamical space-time metrics whose evolution is governed by the Einstein equations. Many approaches, like loop quantum gravity \cite{LQG}, Causal Dynamical Triangulations (CDT) \cite{CDT}, spin foams \cite{SF} and group field theory \cite{GFT}, expect that the quantum theory of gravity becomes manifest at very high energy. That is, quantum analogs of geometry and gravitational properties such as  the quantum Hilbert-Einstein action or Lorentz invariance are built into the high energy theory.  In different theories, various such properties are present at the microscopic level:  In loop quantum gravity, for example, the state space is given in terms of spin network basis states, understood as a discretization of space (see \cite{SnetsLQG,Pen,HasPer}) and, in addition, embedded in 3-dimensional space.  The dynamics of the micro-states is governed by the Hamiltonian constraint, obtained by the canonical quantization of the Hilbert-Einstein action.  In spin foams or group field theory, the microscopic states similarly are based on simplicial graphs or complexes carrying algebraic geometric data, including the Lorentz group.  Depending on the model, a spin foam may be based on states that are embedded complexes or combinatorial (abstract, non-embedded graphs).  In many models, the states carry representations of the Lorentz group, while the dynamics uses the Regge action, a discrete form of the Hilbert-Einstein action.  In CDT, the states are simplicial combinatorial objects, again encoding microscopic Lorentzian geometry, while the dynamics contains the Regge action.  Of course, the important question is how many of these features survive at the low energy or continuum limit.  Is the input of geometric and gravitational properties at the microscopic level necessary for the continuum limit to be gravity?  The answer is at present unclear.  CDT results indicate that while certain properties such as causality are important, others, including the gravitational action, may not be \cite{CDTSHE}.  In spin foams or group field theory, it is not clear that the Lorentz group representations on the microstates ensure the reappearance of the Lorentz group in the continuum limit.  
As an alternative, one could regard quantum gravity as an emergent phenomenon from the low energy theory of a condensed matter system. In this approach, the fundamental theory would have no microscopic degrees of freedom, no gravity, no elements of the Lorentz group, etc.

In recent years, there has been a growing attention to the notion of gravity as an emergent phenomenon \cite{emergentqg}. From Aristotle to Philip Anderson, a long-standing tradition in physics asserts that "The whole is more than the sum of its parts" and that "More is different". The emergent approach is concerned with the study of the macroscopic properties of systems with many bodies. Sometimes, these properties can be tracked down to the properties of the elementary constituents. In recent years, though, there has been a flourishing of novel quantum systems, which show behaviors of the whole system that have no explanation in terms of the constituting particles, but instead of their collective behavior and interaction. When the interaction between the particles cannot be ignored, like in systems of strongly interacting electrons, we see many novel and beautiful properties: gauge fields can emerge as a collective phenomenon, strange quantum phase transitions happen, unusual forms of superconductivity and magnetism appear, novel orders of the matter based on topological properties of the system and featuring exotic statistics are found.

In this sense, one can view the problem of quantum gravity as a problem in statistical physics or condensed matter theory:  we know the low energy physics and are looking for the correct universality class of the microscopic quantum theory.  By analogy to the Ising model for ferromagnetism, one can ask:  What is the ``Ising model'' for gravity?
A number of such approaches to the problem have recently been proposed, ranging from looking for gravity analogues in condensed matter systems \cite{analog,Hu}, to condensed matter systems with emergent graviton excitations \cite{Wengravity}, and spin system models for emergent geometry \cite{KonMarSev,HamMarPreSev}.
 One can also view approaches such as CDT \cite{CDT} or matrix models \cite{MM} in this light.  In comparison to the quantum gravity approaches mentioned above, in the systems we will be investigating there is no straightforward geometric content to the microscopic graph states.  We simply use a dynamical network of quantum relations to describe a world without geometry, where locality is determined by the presence of absence of quantum interactions.  It is perhaps best to think of this system as a quantum information processor, with information more primary than geometry, and geometry being the set of properties such as geometric symmetries that we only expect to find dynamically emergent at low energy.

An important issue in this direction of research in quantum gravity is the dynamical nature of geometry in general relativity.  Normally, methods in condensed matter theory use a fixed background, for example, a spin system on a fixed lattice.   The lattice determines the locality of the interactions and hence is a discretization of geometry.  One can then worry that this direction of research is limited to fixed background geometries.  This is addressed in some, but not all, of the above-mentioned approaches.  For example,
 CDT  is based on path integral dynamics on an ensemble of all lattices (each lattice is a regularization of a Lorenzian geometry), thus providing a proper non-perturbative approach to the problem.  Elsewhere, ideas from quantum information theory have been introduced to deal with this problem \cite{seth,me,HamMarPreSev}.

An alternative direction, and the one we are pursuing in the present article, is to make the lattice itself dynamical.  Two of the present authors proposed such a model for the emergence of geometry in \cite{KonMarSev}.  The basic idea was to promote the lattice links to dynamical quantum degrees of freedom and construct a Hamiltonian such that, at low energy, the system ``freezes'' in a configuration with recognizable geometric symmetries, interpreted as the geometric phase of the model.
The present work revisits the same idea but in a different model with two central properties:
\begin{itemize}
\item
The model is a spin system on a {\em dynamical} lattice.
\item
There are lattice and matter degrees of freedom.
The lattice interacts with the matter:  matter tells geometry how to curve and geometry tells matter where to go.
\end{itemize}

The starting point for the implementation of the above is considerations of locality.  Normally,  locality is specified by the metric $g_{\mu\nu}$ on a manifold $\mathcal M$.  Dynamics of matter on $({\mathcal M}, g_{\mu\nu})$ is given by a Lagrangian which  we call  local if the interaction terms are between systems local according to $g_{\mu\nu}$.
 A  Lagrangian with non-local interaction terms is typically considered unphysical.  That is, the matter dynamics is made to match the given space-time geometry.  We will do the reverse and define geometry via the dynamics of the matter.  Our principle is that if particles $i$ and $j$ interact, they must be adjacent.  This is  a dynamical notion of adjacency in two ways: it is inferred from the dynamics and, being a quantum degree of freedom, it changes dynamically in time. This amounts to a spin system on a dynamical lattice and to interaction of matter with geometry.

To summarize, we present a toy model for the emergence of locality from the dynamics of a quantum many-body system. No notion of space is presupposed. Extension, separateness, distance, and all the spatial notions are emergent from the more fundamental notion of interaction. The locality of interactions is now a consequence of this approach and not a principle.
We will promote the interaction terms between two systems to quantum degrees of freedom, so that the structure of interactions itself becomes a dynamical variable.
This makes possible the interaction and even entanglement between matter and geometry.

This toy model is also a condensed matter system in which the pattern of interaction itself is a quantum degree of freedom instead of being a fixed graph. It can be regarded as a Hubbard model where the strength of the hopping emerges as the mean field value for other quantum degrees of freedom. We show a numerical simulation of the quantum system and results on the asymptotic behavior of the classical system. The numerical simulation is mainly concerned with the entanglement dynamics of the system and the issue of its thermalization as a closed system. A closed system can thermalize in the sense that the partial system shows some typicality, or some relevant observables reach a steady or almost steady value for long times. The issue of thermalization for closed quantum system and the foundations of quantum statistical mechanics gained recently novel interest with the understanding that the role of entanglement plays in it \cite{thermalization}. The behavior of out of equilibrium quantum system under sudden quench, and the approach to equilibrium has been recently the object of study to gain insight in novel and exotic quantum phases like topologically ordered states.

From the point of view of  Quantum Gravity, the interesting question is whether such a system can capture aspects of the dynamics encoded by the Einstein equations.  We start investigating in this direction by studying an analogue of
a trapped surface that may describe, in more complete models, black hole physics. We discuss physical consequences of the entanglement between matter and geometry.

The model presented here is very basic and we do not expect it to yield a realistic description of gravitational phenomena. What we would like to show  is that such a model can have an emergent, quantum-mechanical notion of geometry (even if not smooth), that locality is derivative from dynamics, and the extent to which such a simple model may capture aspects of the Einstein equations, is left to future work.

\section{The Model with hopping bosons on a dynamical lattice}

\subsection{Promoting the edges of the lattice to a quantum degree of freedom}
We start with the primitive notion of a set of $N$ distinguishable physical systems. We assume a quantum mechanical description of such physical systems, given by the set $\{ \mathcal H_i, H_i\}$ of the Hilbert spaces $\mathcal H_i$ and Hamiltonians $H_i$ of the systems $i=1,...,N$. This presumes it makes sense to talk of the time evolution of some observable with support in $\mathcal H_i$ without making any reference to space.

We choose $\mathcal H_i$ to be the Hilbert space of a harmonic oscillator. We denote its creation and destruction operators by $b^\dagger_i,b_i$, respectively, satisfying the usual bosonic relations. Our $N$ physical systems then are $N$ bosonic particles and
the total Hilbert space for the bosons is given by
\be
\mathcal H_{bosons} = \bigotimes_{i=1}^N \mathcal H_i.
\ee
If the harmonic oscillators are not interacting, the total Hamiltonian is trivial:
\be\label{hv}
H_v = \sum_{i=1}^N H_i =- \sum_i \mu_i b^\dagger_i b_i.
\ee
If, instead, the harmonic oscillators are interacting, we need to specify which is interacting with which.  Let us call $\mathcal I$ the set of the pairs of oscillators ${\bf e}\equiv(i,j)$ that are interacting. Then the Hamiltonian would read as
\be
H = \sum_i H_i + \sum _{{\bf e}\in \mathcal I} h_{\bf e}
\label{eq:H}
\ee
where $h_{\bf e}$ is a Hermitian operator on $\mathcal H_i \otimes \mathcal H_j$ representing the interaction between the system $i$ and the system $j$.

We wish to describe space as the system of relations among the physical systems labeled by $i$.  In a discrete setup like ours, a commonly used primitive notion of the spatial configuration of $N$ systems can be provided by an adjacency matrix $A$, the $N\times N$ symmetric matrix defined as follows:
\be
A_{ij}=\left\{{\begin{array}{ll}
1&{\mbox{if $i$ and $j$ are adjacent}}\\
0&{\mbox{otherwise}}.
\end{array}}
\right.
\ee
The matrix $A$ is associated to a graph on $N$ vertices whose edges are specified by its the nonzero entries.  Now, it is clear that the set $\mathcal I$ of interacting nodes in the Hamiltonian (\ref{eq:H}) also defines a graph $G$ whose vertices are the $N$ harmonic oscillators and whose edges are the pairs ${\bf e}\equiv(i,j)$ of interacting oscillators. Here $\mathcal I$ is the edge set of $G$.  We want to promote the interactions - and thus the graph itself - to a quantum degree of freedom.

To this goal, let us define  $\mathcal G$ as the set of graphs $G$ with $N$ vertices. They are all subgraphs of $K_N$, the complete graph on $N$ vertices, whose $\frac{N(N-1)}{2}$ edges correspond to the (unordered) pairs ${\bf e}\equiv(i,j)$ of harmonic oscillators.  To every such pair ${\bf e}$ (an edge of $K_N$)  we associate a Hilbert space $\mathcal H_{\bf e}\simeq \CC^2$ of a spin $1/2$.
 The
total Hilbert space for the graph edges is thus
\be
\mathcal H_{graph} = \bigotimes_{\bf e =1}^{N(N-1)/2} \mathcal H_{\bf e}.
\ee
We choose the basis in $\mathcal H_{graph}$  so that to every graph $g\in\mathcal G$ corresponds a basis element in $\mathcal H_{graph}$: the basis element $\ket{ e_1\ldots  e_{N(N-1)/2}}\equiv\ket{G}$ corresponds to the graph $G$ that has all the edges ${\bf e}_s$ such that $e_s =1$. For every edge $(i,j)$, the corresponding $SU(2)$ generators will be denoted as $S^i = 1/2 \sigma^i$ where $\sigma^i$ are the Pauli matrices.

The total Hilbert space of the theory is
\be
\mathcal{H} = \mathcal H_{bosons}\otimes \mathcal H_{graph},
\ee
and therefore a basis state in $\mathcal H$ has the form
\be
\ket{\Psi} \equiv \ket{\Psi^{(bosons)}}\otimes\ket{\Psi^{(graph)}}\equiv \ket{n_1,...,n_N}\otimes \ket{e_1,...,e_{\frac{N(N-1)}{2}}}
\ee
The first factor tells us how many bosons there are at every site $i$ (in the Fock space representation) and the second factor tells us which pairs ${\bf e}$ interact. That is,  the structure of interactions is now promoted to a quantum degree of freedom. A generic state in our theory will have the form
\be\label{state}
\ket{\Phi} = \sum_{a,b} \alpha_{a,b} \ket{\Psi^{(bosons)}_{a}}\otimes\ket{\Psi_b^{(graph)}},
\ee
with $\sum_{a,b} |\alpha_{a,b}|^2=1$. In general, our quantum state describes a system in a generic superposition of energies of the harmonic oscillators, and of interaction terms among them. A state can thus be a quantum superposition of ``interactions''. For example, consider the systems $i$ and $j$ in the state
\be\label{s1}
\ket{\phi_{ij}} = \frac{\ket{10}\otimes\ket{1}_{ij}+ \ket{10}\otimes\ket{0}_{ij}}{\sqrt{2}}.
\ee
This state describes the system in which there is a particle in $i$ and no particle in $j$, but also there is a quantum superposition between $i$ and $j$ interacting or not. The following state,
\be\label{ent}
\ket{\phi_{ij}} = \frac{\ket{00}\otimes\ket{1}_{ij}+ \ket{11}\otimes\ket{0}_{ij}}{\sqrt{2}}.
\ee
represents a different superposition, in which the particle degrees of freedom and the graph degrees of freedom are entangled. It is a significant feature of our model that {matter can be entangled with geometry}.

An interesting interaction term is the one that describes the physical process in which a quantum in the oscillator $i$ is destroyed and one in the oscillator $j$ is created. The possibility of this dynamical process means there is an edge between $i$ and $j$. Such dynamics is described by a Hamiltonian of the form
\be\label{hhop}
H_{hop} = -t\sum_{(i,j)}P_{ij}\otimes (b^\dagger_ib_j +b_i b^\dagger_j)
\ee
where
\be
P_{ij}\equiv S^+_{(i,j)}S^-_{(i,j)} =\ket{1}\bra{1}_{(i,j)} = \left( \frac{1}{2}-S^z \right)_{(i,j)}
\ee
is the projector on the state such that the edge $(i,j)$ is present and the spin operators are defined as $S^+_{(i,j)} = \ket{1}\bra{0}_{(i,j)}$ and $S^-_{(i,j)} = \ket{0}\bra{1}_{(i,j)}$. With this Hamiltonian, the state Eq.(\ref{s1}) can be interpreted as the quantum superposition of a particle that may hop or not from one site to another. It is possible to design such systems in the laboratory. For instance, one can use arrays of Josephson junctions whose interaction is mediated by a quantum dot with two levels.

We note that it is the
dynamics of the particles described by $H_{hop}$ that gives to the degree of freedom $\ket{e}$ the meaning of geometry\footnote{
We use {\em geometry} and {\em space} as shorthand for the adjacency relations encoded in $|\Psi^{graph}\rangle$, even though the generic $|\Psi^{graph}\rangle$ will be a graph (or a superposition of graphs) without any symmetries and hence not a candidate for a discretization of a smooth geometry.  In this simple model, we make no attempt to dynamically flow to a $|\Psi^{graph}\rangle$ with recognizable geometric symmetries, as was done, for example, in \cite{Kon,MarSmo,KonMarSev}.  One can address this in a future model by extending the Hamiltonian of the model.
}.
 The geometry at a given instance is given by the set of relations describing the dynamical potentiality for a hopping. Two points $j$, $k$ can be "empty", that is, the oscillators $j,k$ are in the ground state, but they can have a spatial relationship consisting in the fact that they can interact. For example, they can serve to have a particle to hop from $i$ to $j$, then to $k$, then to $l$. We read out the structure of the graph from the interactions, not from the mutual positions of particles.

In addition, $H_{hop}$
  tells us that it takes a finite amount of time to go from $i$ to $j$. If the graph is represented by a chain, it tells us that it takes a finite amount of time (modulo exponential decaying terms) for a particle to go from one end of the chain to another. This results to a ``spacetime'' picture (the evolution of the adjacency graph in time) with a finite lightcone structure. The hopping amplitude is given by $t$, and therefore all the bosons have the same speed. We can make the model more sophisticated by enlarging the Hilbert space of the links, and obtain different speeds for the bosons. Instead of considering spins $1/2$, consider an $S-$level system. The local Hilbert space is therefore
\be
\mathcal H_e = \mbox{span} \{ \ket{0}, \ket{1},\ldots, \ket{S-1} \}
\ee
Now consider the projector onto the $s-$th state on the link $(i,j)$: $P^{(s)}_{ij} = \ket{s}\bra{s}_{ij}$. We can define a new hopping term whose amplitude depends on the level of the local system in the following way:
\be\label{hhop2}
H_{hop} = -\sum_{s,(i,j)} t_s P^{(s)}_{ij}\otimes (b^\dagger_ib_j +b_i b^\dagger_j)
\ee
where the hopping amplitudes $t_s$ depend on the state $s$ of the system, and $t_0=0$. For instance, the $t_s$ can be chosen larger for larger $s$. In this way, moves through higher level links are more probable, and therefore the speed of the particles is not constant. In the following, we will study the model with just the two level system.

Of course, we need a Hamiltonian also for the spatial degrees of freedom alone. The simplest choice is simply to assign some energy to every edge:
\be\label{hlink}
 H_{link} = -U\sum_{(i,j)} \sigma^z_{(i,j)}
\ee

Finally, we want space and matter to interact in a way that they can be converted one into another. The term
\be
H_{ex} =k \sum_{(i,j)}  \left( S^-_{(i,j)}\otimes (b^\dagger_i
b^\dagger_j)^R +S^+_{(i,j)}\otimes (b_i  b_j)^R\right)
\ee
can destroy an edge $(i,j)$ and create $R$ quanta at $i$ and $R$ quanta at $j$, or, vice-versa, destroy $R$ quanta at $i$ and $R$ quanta at $j$ to convert them into an
edge.

The terms $H_{link}$ and $H_{ex}$ are so simple that we will not expect them to give us any really interesting property of how regular geometry can emerge in such a system. This is the subject for a more refined and future work. Nevertheless, this term has an important meaning because the nature of the spatial degrees of freedom is completely reduced to that of
the quanta of the oscillators: an edge is the bound state of $2R$ quanta. When in the edge form, the quanta cannot hop around. When unbounded, they can hop around under the condition
that there are edges from one vertex to another.
One can replace the separation of the fundamental degrees of freedom into bosons and graph edges with a unified set of underlying particle ones, single bosons and collections of $2R$ bound bosons.
Therefore a bound state of $2R$ quanta in the pair $(\mathcal H_i, \mathcal H_j)$ tells us what physical systems are at graph distance one. The set of such bound states as we vary $j$ is the neighborhood of the system $i$. This is the set of  vertices $j$ a free particle in $i$ can hop to. The projector $P_{ij}$ has thus the meaning that the hopping interaction must be {\em local} in the sense just defined.

Now we see that the term $H_{ex}$ is not satisfactory because exchange interactions are possible between any pair of vertices, no matter their distance. So quanta that are far apart can be converted in an edge between two points that were very far just before the conversion. Moreover, also the conjugate process is problematic, because it can easily lead to a graph made of disconnected parts. We implement locality by allowing exchange processes only between points that are connected by some {\em other} short path of length $L$. Note that this refers to the locality of the state $|\Psi^{graph}\rangle$ at time $t$ relative to the locality of the state at time $t-1$.
Consider again the projector $P_{ij}$ on the edge $(i,j)$ being present. Its $L-$th power is given by
\be
P_{ij}^L =\sum_{k_1,...,k_{l-1}} P_{ik_1}P_{k_2k_3}\dots P_{k_{L-1}j}.
\ee
For every state $\ket{\Psi}\in\mathcal H_{graph}$, we have that $P_{ij}\ket{\Psi}\ne {\bf 0}$ if and only if there is at least another path of length $L$ between $i$ and $j$.
We can now modify the term $H_{ex}$ as follows:
\be\label{hex}
H_{ex} =k \sum_{(i,j)}  \left(S^-_{(i,j)}P^L_{ij}\otimes (b^\dagger_i
b^\dagger_j)^R +P^L_{ij}S^+_{(i,j)}\otimes (b_i  b_j)^R\right).
\ee
In the extended $S-$level system, the exchange term is modified as $\ket{(s+1) (\mbox{mod} S)}\bra{s(\mbox{mod} S)}_{(i,j)} \otimes(b_i b_j)^{R}$ and similarly for the hermitian conjugate.

This was the final step that brings us to the total Hamiltonian for the model which is
\be \label{h}
H = H_{link}+ H_v + H_{ex} + H_{hop}.
\ee
In the following, we consider the theory for $L=2$, which is the strictest notion of locality  for the exchange interaction one can implement.

\subsection{Discussion of the model}
We can summarize the model in the following way. All we have is matter, namely the value of a function $f_i$, where the indexes $i$ label different physical systems. We have chosen  $f_i$ to be the number of quanta of the $i$-th harmonic oscillator. The bound state of a particle in $i$ and a particle in $j$ has the physical effect that other particles in $i$ and $j$ can interact.  When there is such a bound state, we say there is an edge between $i$ and $j$. Then other particles at $i$ and $j$ can interact, for instance, they can hop from $i$ to $j$.  The collection of these edges, or bound states, defines a graph which we interpret as the coding of the spatial adjacency of the particles (in a discrete and relational fashion).
The physical state of the many body system is the quantum superposition of configurations of the particles and of the edges. The system evolves unitarily, and particles can hop around along the edges. But the distribution of the particles also influences the edges because some particles at vertices $i,j$ can be destroyed (if $i$ and $j$ are nearby in the graph) to form another edge, and therefore making $i$ and $j$ nearer. The new edge configurations then influence the motion of the particles and so on. We have a theory of matter interacting with space. The intention of the model is to study to what extend such dynamics captures aspects of the Einstein equation and whether it (or a later extension of such a spin system) can be considered as a precursor of the gravitational force. From the condensed matter point of view, this is a Hubbard model for hopping bosons, where the underlying graph of the Hubbard model is itself a quantum dynamical variable that depends on the motion of the bosons. In the spirit of General Relativity, the edges (space) tell the bosons (matter) where to go, and the bosons, by creating edges, tell the space how to curve.

We note that, in this theory, all that interacts has a local interaction by definition. We defined locality using the notion of neighborhood given by the set of systems interacting with a given system \footnote{
There is, however, a way to define non-local interactions. For instance, consider the configuration of the graph of the square lattice. Pick two vertices $i,j$ at large distance $l_{ij}$ on this lattice (in the graph distance sense), and place a new edge connecting them. Now by definition the two vertices $i$ and $j$ are adjacent. Nevertheless, the ratio between the number of paths of distance one from $i$ to $j$  with the number of paths of distance $l_{ij}$ goes to zero in the limit of large $l_{ij}$. This is what one can call a {\em non-local} interaction. It corresponds to the situation in lattice field theory where one has a fixed graph, and then some interactions between points that are far apart.}.  We also note that, due to quantum superpositions, matter and space can be entangled. For this reason, the dynamics of the matter alone is the ruled by a quantum open system, the evolution for the matter degrees of freedom is described no more by a unitary evolution operator but by a completely positive map. We can show that the entanglement increases with the curvature. To fix the ideas, let us start with a flat geometry represented by the square lattice as the natural discretization of a two dimensional real flat manifold. In this model, a flat geometry with low density of matter can be described by a square (or cubic) lattice with a low density of bosons. This means that a particle is most of the time alone in a region that is a square lattice. The model will not then allow interaction between the particle and the edges, and all that happens is a free walk on the graph. On the other hand, when we increase the degree of the vertices by adding more edges, we make interaction, and hence entanglement, between edges and particles possible. This corresponds to increasing the curvature. In a regime of very weak coupling, $k\ll t\ll U,\mu$, entanglement will be possible only in presence of extremely strong curvature. From the point of view of the dynamics of the quantum system, this means that the evolution for the matter is very close to be unitary when curvature is low, while very strong curvature makes the evolution for the particles non-unitary and there will be decoherence and dissipation with respect to the spatial degrees of freedom.

How does the graph evolve in time in such a model? The quantum evolution is complex, and since the model is not exactly solvable, numerical study is constrained to very small systems. In the next section we simulate the system with $4$ vertices and hard core bosons.

We can gain some insight from the analysis of the classical model, regarding $H$ as the classical energy for classical variables. Since we delete edges randomly and build new edges as the result of a random walk of the particles, and there is nothing in this model that favors some geometry instead of others, we do not expect to obtain more than random graphs in the limit of extremely long times. Indeed, we can argue as follows. With the exception of a very small number or graphs, all the other graphs belong to the set of graphs in which one can - under the evolution of our model - reach a ring. In practice, this means that there is a configuration in which one deletes all the edges without disconnecting the graph, and obtains many particles. This means that, starting from a state with zero particles, $N$ vertices and $L_0$ edges the number of eligible edges for deletion is $L_0 - \alpha N$ with some constant $\alpha$ of order $1$. So, as long as $L_0 > \alpha N$, the dynamical equilibrium between the number of edges and particles is realized when the rate of conversion of edges into particles equals the one of conversion of particles into edges. Let $r$ be the number of edges destroyed (and pairs of particles created). If we assume that the particles move much faster than the edges, they will always be eligible for creating a new edge. The rates will be then the same when $L_0-r = r$ which implies that at the dynamical equilibrium half of the initial edges are destroyed. Consider the case of the complete graph $K_N$, with zero initial particles. The initial number of edges is $L_0 = N(N-1)/2$, and at long times, the dynamical equilibrium is reached when  $r = L_0/2$. A similar reasoning can be applied to r instance, consider a square lattice of $N$ vertices with periodic boundary conditions. In the state without particles, this is an eigenstate of the Hamiltonian, and therefore its evolution is completely frozen. Nevertheless, if we add a pair of particles, then all the other configurations of the graph can be reached, including the ring immersed in a gas of many particles. It turns out that the dynamical equilibrium is reached when $L_0/2 =N$ of the initial edges are destroyed. The equilibrium state is obtained by deleting the edges randomly, and thus we expect to obtain a random graph. In order to obtain more interesting stable geometries, one has to put other terms in the Hamiltonian, that involve more edges together, meaning that curvature has a dynamical importance. The rigorous treatment of the asymptotic evolution of such graphs requires an analysis in terms of Markov chains, and it is developed in the next section.

\subsection{Trapped surfaces}
Another feature of this model is that allows for a very good trapping of particles and light, Let us consider a configuration of the system in which we have a region $\mathcal S$ of the graph with $N_S$ vertices, that is highly connected, in an almost complete way, and then connected to the rest of the graph that can be a square or cubic lattice. We assume to be in the large limit for the ratio $r$ between the number of edges in $\mathcal S$ and the number of edges between connecting $\mathcal S$ to the rest of the graph. Initially there are no particles. The region of the graph with "flat" geometry, that is, the square lattice, is frozen. In the highly connected lump, edges will start converting in particles. We have two time scales. The time scale to reach the dynamical equilibrium between the processes of destruction/construction of edges, and the time scale for particles to escape from $\mathcal S$. We assume that $r$ is large enough that the dynamical equilibrium is reached well before $\mathcal S$ starts losing particles. As a matter of fact, the region $\mathcal S$ will behave as the complete graph. Half of the edges will convert in particles, and we will have a dynamical equilibrium between a graph with $O(N_S^2)$ edges in a gas of $O(N_S^2)$ particles. Other particles coming from outside, will get trapped inside $\mathcal S$ too, therefore increasing the number of edges and particles within it. We want now to show that not even light can escape from this region. First of all, we want to have light propagating in our model. This can be done by adding to the Hamiltonian the term of Wen's $U(1)$ theory for emerging light \cite{wenlight}. In the phase where the couplings of the $U(1)$ theory are small with respect the others, we can have electromagnetic waves traveling on the graph. Now consider the state with very low density of particles, and with a graph that is represented by the wave function $\Psi = \Psi _A \otimes \Psi_B$. Here $B$ is a set of $n_B$ nodes such that $n_B \ll n_A$. The wavefunction $\Psi_A$ is chosen to be the one representing a 3D cubic lattice. The wavefunction $\Psi_B$ is instead chosen to be a very high dimensional hypercubic lattice, with dimension $D\sim n_B$. Then one has to knit carefully the nodes in $A$ with some of those in $B$. We choose to knit them with just $n^\frac{2}{D}_B$ of them, representing their "surface" $\Sigma_B$. To summarize, the wave-function in the region outside the surface represents a discrete version of an euclidean three dimensional space. Inside the surface, we still have an euclidean space, but of very high dimension. On such a lattice, which has a well defined geometry, the emergent light is obeys the laws of geometric optics. On the other hand, a general wave-function $\Phi$ for the edges would not posses any definite geometric meaning and the Wen's model cannot even be defined. 

The speed of light $c$ on this graph can be estimated using the Lieb-Robinson bounds \cite{LRS} and one can prove that $c$ has different values in the different mediums $A,B$ and it is proportional the geometric dimension of the medium so that we have $c_B/c_A \sim n_B$. The argument about the asymptotic states shows that the number of edges in $B$ will always be of the order of $n_B$. In the geometry we have chosen, it makes sense to speak about the angle of a ray of light with respect to the normal to the surface separating $A$ and $B$. The Snell's law of optics will imply that the critical angle for total internal reflection is
\be
\theta_c = \sin^{-1} \frac{c_B}{c_A} \simeq \sin^{-1} n_B
\ee
and therefore the probability of an emerging ray of light from $B$ is of the order of $n^{-1}_B$. We see that light (and matter) are trapped within the hypersurface in the large $n_B$ limit. Eventually, some light can come out, and since the outer graph, even though less dense, contains more vertices and edges, eventually the region of space $B$ has to evaporate. The whole process is completely unitary. Nevertheless, the emitted quanta of light and matter are entangled with the spatial degrees of freedom  (the edges) inside $\Sigma_B$. But when $B$ has evaporated, there are still edges there, with a density not much different from those outside $\Sigma_B$. Therefore the final state of the emitted quanta can be still entangled with the degrees of freedom inside $\Sigma_B$ even if the "black hole" has evaporated. The spectrum of the emitted radiation obtained by tracing out the spins in $\Sigma_B$ will therefore be mixed, even though the whole process is unitary. The black hole information paradox can be described in terms of entanglement. Pairs of particles inside/outside the event horizon are created and these pairs are entangled. The density matrix of the particles outside the horizon is therefore mixed, because there is a classical mixture of the particles coming from having to trace out all the degrees of freedom inside the horizon to which we have no access. The problem is that when the black hole disappears, there is nothing for the particles to be entangled with, and the mixture becomes a paradox. In our model, the particles are entangled with the spatial degrees of freedom.  The disappearance of the black hole just means that the spatial degrees of freedom acquire a particular configuration, but the particles are still entangled with them, as in Eq.(\ref{ent}).


\section{The model with hard core bosons}
\begin{figure}
  \includegraphics[width=14cm]{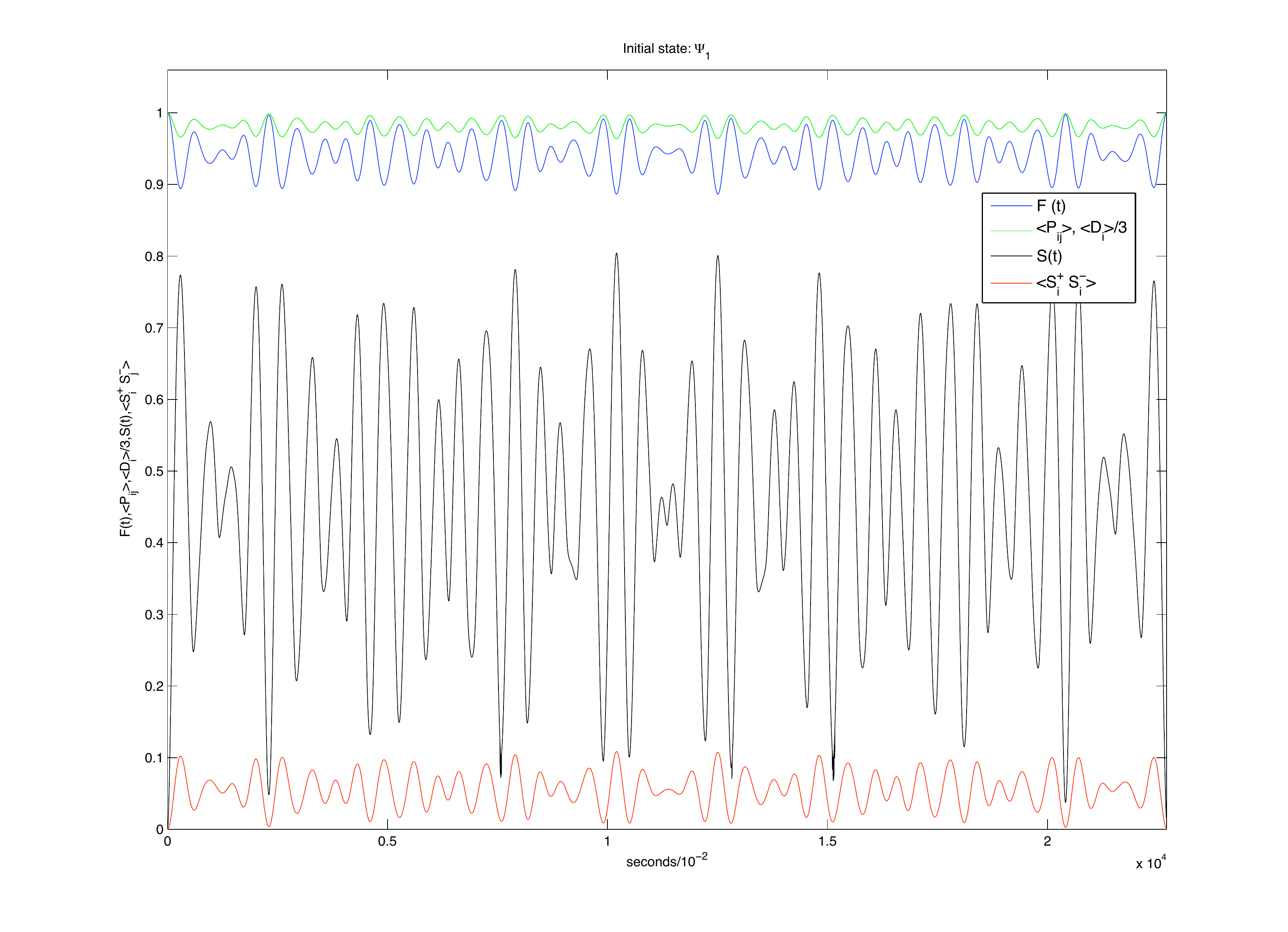}\\
  \caption{Simulation of the system $H_{1/2}$ for $N=4$. The initial state is $\rho_1 = \ket{\psi_1}\bra{\psi_1}$ where $\ket{\psi_1} = \ket{0000111111}$, that is, there are no particles and all edges are present. The parameters for this simulation are $U= \mu  =1, t=k =.1$. In the figure are plotted the quantities $\langle S^+_i S^-_i \rangle$ (red line), $\mathcal F (t)$ (blue line), $S(t)$ (black line), $P_{ij}(t), D_i(t)/3 $ (green line) as a function of time.  Revivals of the expectation value of the link operator coincide with revivals in the fidelity with the initial state. The initial value of the entanglement is $S(0)=0$ because the initial state is separable. Notice that even though the fidelity is $\mathcal F (t) \gtrsim 0.85$, the state has a non negligible entanglement.}
  \label{q1}
\end{figure}

\begin{figure}
  \includegraphics[width=18cm]{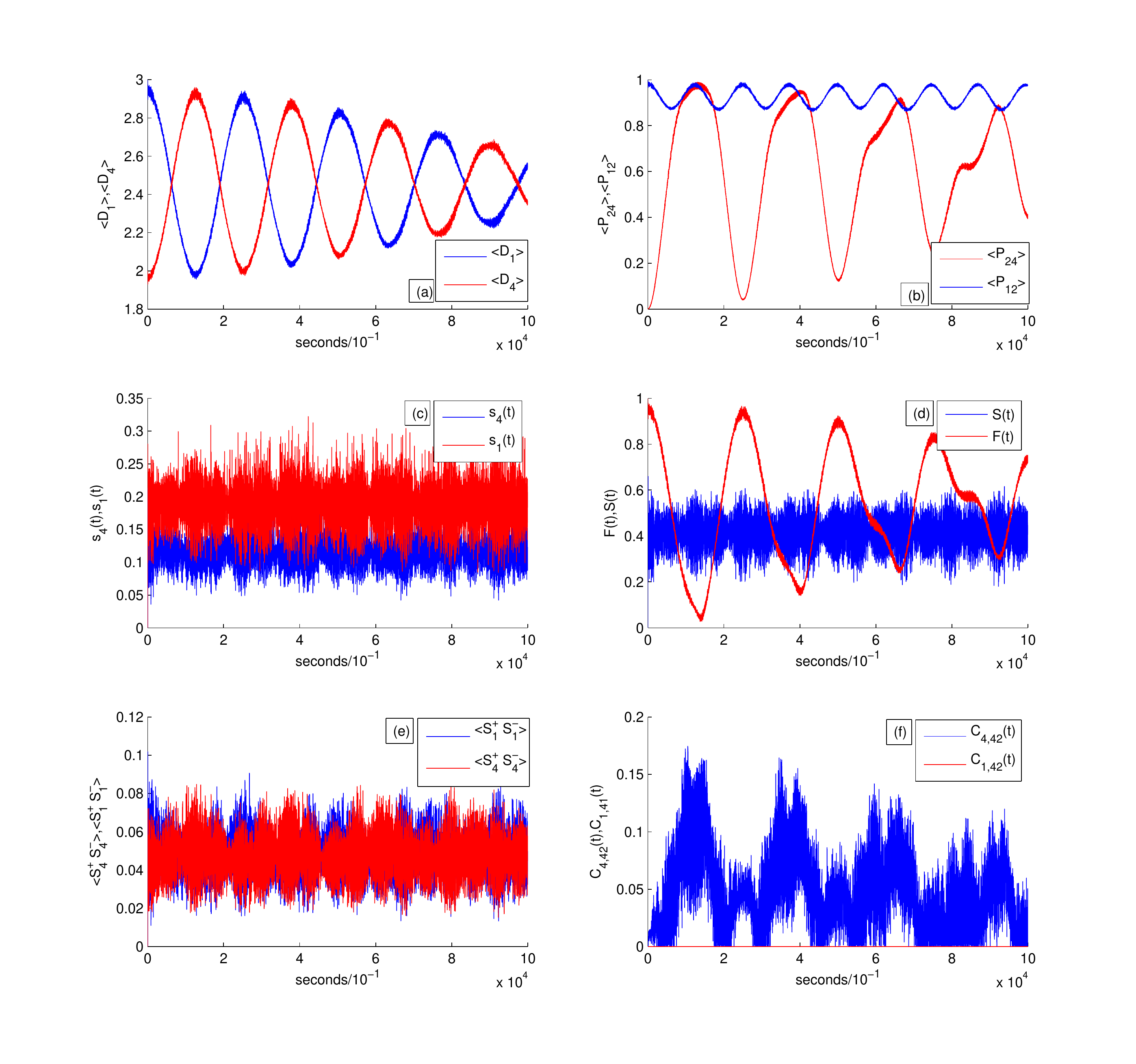}\\
  \caption{The initial state of this numerical integration is $\ket{\psi_2}=\ket{0000111101}$, the parameters are in the ``insulator" phase, $U=\mu=1$, $t=k=0.1$. The temporal scale is in units of $\hbar$, on a range of $10^4$ seconds. The diagonalization of the full system has been performed by means of Householder reduction. \textbf{a)} Time evolution of $\langle D_1 (t)\rangle, \langle D_4 (t)\rangle$. The damping of the oscillations is a sign of thermalization. \textbf{b)} Expectation values $\langle P_{12}(t)\rangle, \langle P_{24}(t)\rangle$. The latter observable is thermalizing. \textbf{c)} Von Neumann Entropy $s_i (t)$ for the sites $i = 1,2$. We see that the entanglement dynamics is split in two different bands. The two vertices are only distinguished by the initial degree. \textbf{d)} Entanglement evolution $S(t)$ and overlap with the initial state $\mathcal F (t)$.  The damping of $\mathcal F (t)$ is a clear sign of thermalization. The entanglement $S(t)$ between particles and edges shows the entangling power of the system. \textbf{e)} Expectation value of the particle operators at two different sites $i=1,4$. \textbf{f)} Concurrence $C(t)$ as a function of time of the particles on the site $i=2$ with the edge $(2,4) (blue)$. Again we notice a damping of oscillations. Instead, the concurrence between the site 1 and the link 5  (red) is identically zero.}
  \label{q2}
\end{figure}

\begin{figure}
  \includegraphics[width=18cm]{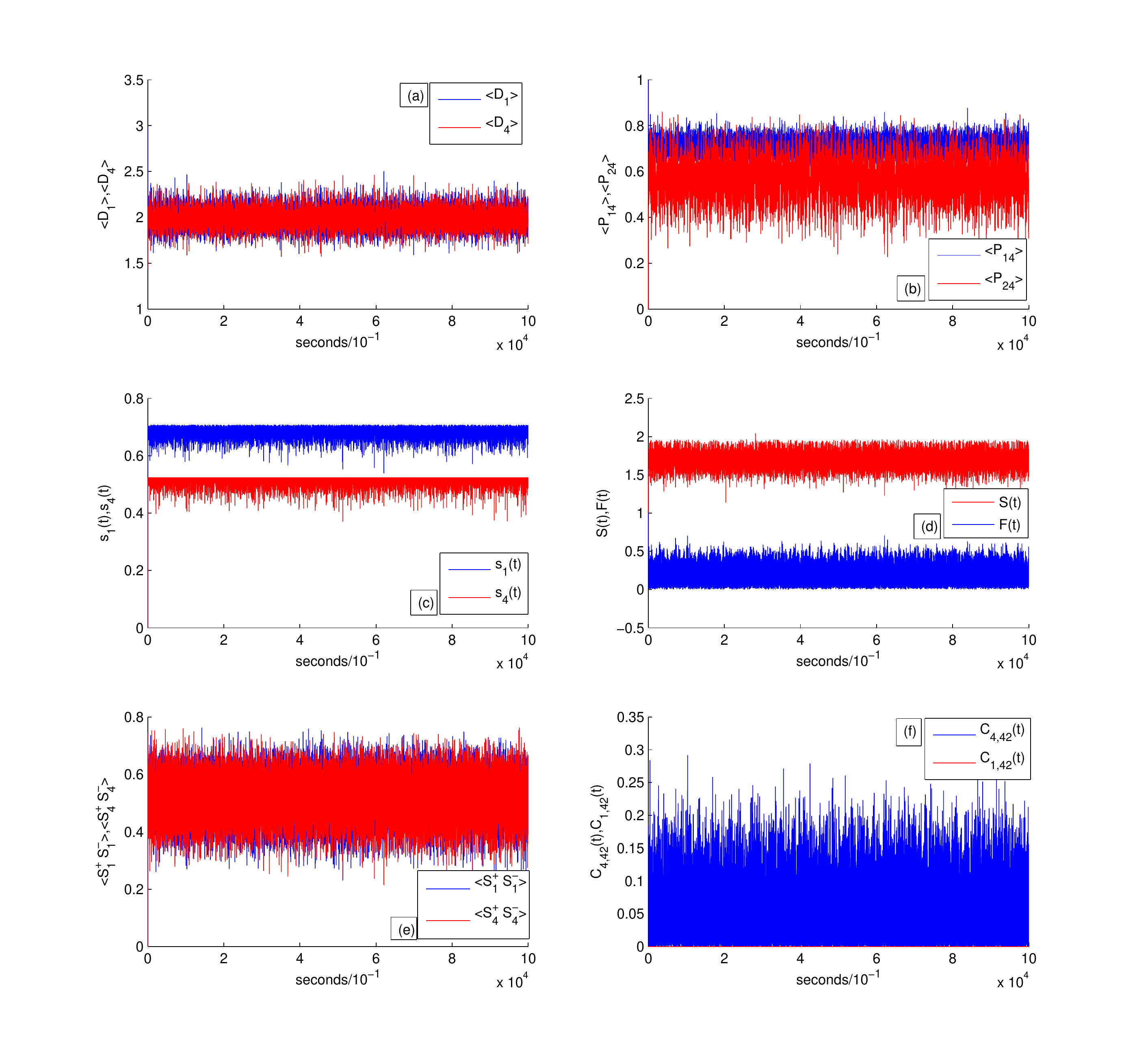}\\
  \caption{ The initial state of this numerical integration is $\ket{\psi_2}=\ket{0000111101}$, the parameters are in the ``superfluid" phase, $U=\mu=1$, $t=k=0.1$. The temporal scale is in units of $\hbar$, on a range of $10^4$ seconds. The diagonalization of the full system has been performed by means of Householder reduction.  \textbf{a)} Time evolution of $\langle D_1 (t)\rangle, \langle D_4 (t)\rangle$.. The oscillations have constant amplitude and the system does not present signs of thermalization. \textbf{b)} Expectation values of the link operators $ P_{14}, P_{24}$. There is no sign of thermalization. The two values belong to two different bands depending on the initial value of the operator. \textbf{c)} Von Neumann Entropy $s_i (t)$ for the sites $i = 1,2$. We see that the entanglement dynamics is split in two different bands. The two vertices are only distinguished by the initial degree and the splitting is more marked than in the ``insulator" case. Compare the result with the higher overlap of the operators $D_{ij} (t)$. \textbf{d)} Entanglement evolution $S(t)$ and overlap with the initial state $\mathcal F (t)$.  Again the plots show no signs of thermalization. The behavior of $\mathcal F(t)$ implies very long recurrence times. \textbf{e)} Expectation value of the particle operators at two different sites $i=1,4$. \textbf{f)} Time evolution of Concurrence $C(t)$. In blue is plotted the Concurrence between the vertex $i=2$ and the edge $(2,4)$ for the ``superfluid" case.  Unlike the insulator case, the behavior of $C(t)$ does not show any sign of thermalization. The concurrence between the site 1 and the link 5 is identically zero, as in the ``insulator'' case.}
  \label{q3}
\end{figure}

\subsection{Setting of the model}
In this section, we study the model Eq.~(\ref{h}) when the particles are hard core bosons. In this model, only at most one particle is allowed per site and the model can be mapped onto a spin system. We are particularly interested in the entanglement dynamics of the system. We have performed a numerical simulation of the time evolution of the model described by Eq. (\ref{h}). Since we are interested also in describing the quantum correlations in the reduced density matrix, we have resorted to exact diagonalization. In this way, we are able to compute the entanglement of the matter degrees of freedom with respect to the spatial ones. Of course, the simulation of a full quantum system is heavily constrained by the exponential growth of the Hilbert space. In this work, we have resorted to the simulation of hard-core bosons: at most one particle is allowed at any site. Hardcore bosons creation and annihilation operators must thus satisfy the constraints
\ba
(\hat{b}_i^\dagger)^2=(\hat{b}_i)^2=0 \\
\{ \hat{b}_i, \hat{b}^\dagger _i \} = 1
\ea
With these constraints, the bosonic operators map into the $SU(2)$ generators
\ba
\hat{b}_i^\dagger && \leftrightarrow S^+_i \\
\hat{b}_i && \leftrightarrow S^-_i \\
\hat{b}_i^\dagger \hat{b}_i && \leftrightarrow\left(  \frac{1}{2} - S^z \right)_i
\ea
The local Hilbert space of a site $i$ for a hard core boson is therefore that of a spin one half: $\mathcal H_i^{hcb} \simeq \CC^2$. After the projection onto the hard-core bosons subspace, the model becomes a purely spin $1/2$ model. For a system with $n$ sites, the Hilbert space for the particles is thus the $2^N$-dimensional Hilbert space $\mathcal H_{bosons} = \otimes_{i=1}^N \mathcal H_{i}^{hcb}\simeq \CC^{2\otimes N}$. The Hilbert space for the spatial degrees of freedom is still the $2^{N(N-1)/2}$-dimensional Hilbert space
$\mathcal H_{graph} = \bigotimes_{\bf e =1}^{N(N-1)/2} \mathcal H_{\bf e}$. The total Hilbert space is thus the $2^{N(N+1)/2}$-dimensional Hilbert space
\be
\mathcal{H}_{spins} = \mathcal H_{bosons}\otimes \mathcal H_{graph}\simeq \bigotimes_{i=1}^N \mathcal H_{i}^{hcb}\bigotimes_{\bf e =1}^{N(N-1)/2} \mathcal H_{\bf e}
\ee
As a basis for $\mathcal{H}_{spins}$ we use the computational basis. The basis is thus $\{ \ket{i_1, ..., i_{N(N-1)/2}; j_1,...,j_N}\}$, where the first $N(N-1)/2$ indices $i_k$ label the edges of the graph, and the remaining $N$ indices $j_k$ label the vertices. Of course $i_k, j_k = 0,1$ for every $k$.

After the projection onto the hard core bosons space $\mathcal{H}_{spins}$, the model Hamiltonian becomes thus the spin one-half Hamiltonian (for $\mu_i$ uniform):
\ba\label{onehalf}
\nonumber
H_{1/2} =&&-U\sum_{(i,j)} S^z_{(i,j)}-\mu\sum_{i=1}^N \left(  \frac{1}{2} - S^z \right)_i- t\sum_{(i,j)}  P_{ij}\otimes (S^+_i S^-_j +S^-_i S^+_j)\\
&& - k \sum_{(i,j)}  \left(S^-_{(i,j)}P^{2}_{ij}\otimes (S^+_i S^+_j) +P^{2}_{ij} S^+_{(i,j)}\otimes (S^-_i  S^-_j)\right)
\ea
Let us examine the model in some limits. When the exchange term is vanishing, $k=0$, the model has particle number conservation
\be
[H_{1/2},\hat{N}]=0, \qquad \hat{N} = \sum_i b^\dagger_ib_i
\ee
and therefore it has a $U(1)$ symmetry, corresponding to the local transformation at every site given by
\be
\ket{\psi}\rightarrow \prod_l e^{i\phi b^\dagger_lb_l}\ket{\psi}, \qquad \phi\in[0,2\pi)
\ee
while the total system with $k\ne 0$ does not have particle conservation because particles can be created or destroyed by means of the exchange term with the edges. 
Moreover, the $k=0$ system is self dual at $\mu=0$ under the transformation $b_i\rightarrow b^\dagger_i$.  For every separable state of the form $\ket{\psi} =\ket{i_1, ..., i_{N(N-1)/2}} \otimes \ket{\psi}_{bosons} $, the system is just the usual Hubbard model on the graph specified by the basis state $\ket{i_1, ..., i_{N(N-1)/2}}$.
In the limit of $-U$ positive and very large, all the edges degrees of freedom are frozen in the $\ket{1}$ state. The model becomes a Bose-Hubbard model for hard-core bosons on a complete graph.

It is a typical feature of the richness of the Hubbard model that summing the potential and kinetic term gives a model with an incredibly rich physics. Depending on the interplay between potential and kinetic terms, it can accommodate metal-insulator transitions, ferromagnetism and antiferromagnetism, superconductivity and other important phenomena. The richness of the model comes from the interplay between wave and particle properties.
The hopping term describes degrees of freedom that behave as 'waves', whereas the potential term describes particles \cite{tasaki}. As it is well known, the model is not solvable in two dimensions. The present model is even more complicated by the fact that the graph itself is a quantum variable. It is therefore extremely difficult to extract results from such a model. The hopping term in $t$ favors delocalization of the bosons in the ground state, while the chemical potential $\mu$ is responsible for a finite value of the bosonic density $\rho$ in the ground state given by
\be
\rho = \frac{1}{N}\sum_i \langle b^\dagger_i b_i \rangle.
\ee
The strength of $|\mu|$ determines how many bosons are present in the ground state. For $\mu >0$, a large value of $\mu$ determines $\rho=1$, meaning that the ground state has a boson at every site, whereas for $\mu <0$, a large value of $\mu$ means there are no bosons in the ground state $\rho=0$. In any case, there is no possibility for hopping and this situation describe what is called a Mott insulator. On the other hand, for $k=0$ and $t>\mu$ the hopping dominates and the system is in a superfluid phase. The non vanishing expectation value in the ground state is that of the average hopping amplitude per link
\be
\sigma =\frac{2}{N(N-1)} \sum_{i,j}\langle b^\dagger_i b_j\rangle.
\ee
We expect this situation to hold even for the weakly interacting system $t\gg k \ne 0$. As in the Hubbard model, there should be a quantum phase transition between the Mott insulator and the superfluid phase for a critical value of $\mu/ t$. An extensive numerical simulation of the ground state properties of the model is necessary to understand if, for $k\ne 0$, such transition belongs to the same universality class or a different one. It would also be interesting to understand whether there is a Lieb-Mattis theorem for such a system, namely that there are gapless excitations in the thermodynamic limit for the system of spins one-half.

It should also be evident that depending on the interplay between potential and kinetic energy, the ground state of the system is entangled in the bipartition edges-particles. Starting instead from some separable initial state, the unitary evolution induced by $H_{1/2}$ will entangle states initially separable.

\subsection{Numerical analysis}
We have analyzed several aspect of the dynamics of the system in two different situations. The ``insulator" case is the one in which the potential energies are dominant over the kinetic terms: $U=\mu=1; t=k= 0.1$. The second situation is when the kinetic terms are much stronger, the so called ``superfluid" case: $U=\mu=0.1 ; t=k=1$. We have studied numerically the entanglement dynamics of the model, using as figures of merit the {\em (i)} Entanglement between particles and edges expressed by the von Neumann entropy $S(t)$ of the density matrix reduced to the system of the particles, {\em (ii)} the Entanglement per site $j$ expressed by the von Neumann entropy $s_j(t)$ of the density matrix reduced to the system of just one site, and {\em (iii)} the Concurrence $C(t)$ between a pair of edges, or particles or the particle-edge pair. This expresses the entanglement between these two degrees of freedom alone.

We have simulated the system described by  $H_{1/2}$ with $N=4$ sites, which is $2^{10}$ dimensional. We have labeled the sites $i =1,..,4$ starting from the lower left corner of a square and going clockwise. The basis states for the system are 
$\ket{J_1 J_2 J_3 J_4 ; e_{14} e_{12} e_{23} e_{34} e_{24} e_{13} }$ with $J_i, e_{kl} = 0,1$.  
By direct diagonalization of the Hamiltonian, we compute the time evolution operator $U(t) = e^{-iHt}$. Starting from an initial state $\rho (0)$, the evolved state is $\rho(t) = U(t) \rho U^\dagger (t)$. The entanglement $S(t)$ as a function of time between particles and edges is obtained by tracing out the spatial degrees of freedom, we obtain the reduced density matrix for the hard core bosons: $\rho_{hcb}(t) = \mbox{Tr}_{graph}\rho (t)$. The evolution for the subsystem is not unitary but described by a completely positive map. The entanglement is computed by means of the von Neumann entropy for the bipartition $\mathcal{H} = \mathcal H_{bosons}\otimes \mathcal H_{graph}$, so we have
\be\label{entanglement}
S(t) = -\mbox{Tr} \left(\rho_{hcb}(t)\log  \rho_{hcb}(t) \right)
\ee
The single-site entanglement $s_j(t)$ is instead obtained by tracing out all the degrees of freedom but the site $j$ and then computing the von Neumann entropy of such reduced density matrix. Finally, the last  figure of merit to describe the entanglement dynamics of the model is  the two-spins concurrence $C(t)$ defined in \cite{wot}. We define the $\tau (t)$ reduced system of any two spins in the model, i.e., an edge-edge pair, or an edge-particle pair or a particle-particle pair. The entanglement as function of time between the two members of the pair is given by
\be
\mathcal{C}(\tau(t))\equiv\max(0,\sqrt{\lambda_1}-\sqrt{\lambda_2}-\sqrt{\lambda_3}-\sqrt{\lambda_4}),
\ee
where $\lambda_i$'s are the eigenvalues (in decreasing order $\lambda_1 >\lambda_2 >\lambda_3 >\lambda_4$ ) of the operator $\tau (t)(\sigma_{y}\otimes\sigma_{y})\tau^{*}(t)(\sigma_{y}\otimes\sigma_{y}) $.

There are other important quantity to understand the time evolution of the model. We have computed the expectation value of the particle number operator $S^+_iS^-_i = \ket{1}\bra{1}_i$ at the site $i$, the link operator $P_{ij}$, and the vertex degree operator $D_i = \sum_{k\ne i} P_{ik}$, whose expectation value gives the expected value for the number of edges connected to the vertex $i$. The last important quantity is the fidelity $\mathcal F (t) := |\langle \psi(0)| \psi(t)\rangle|$  of the state $\ket{\psi(t)}$ with the initial state $\ket{\psi(0)}$. This quantity gives a measure of how much the state at the time $t$ is similar to the initial state.

The simulations have been carried out using two initial states $\ket{\psi_1},\ket{ \psi_2}$. The state $\ket{\psi_1}$ is the basis state describing the complete graph $K_4$ without particles: $\ket{\psi_1}=\ket{0000111111}$.  In Fig. \ref{q1} is shown the result of the simulation using $\ket{\psi_1}$ as initial state, and for the model where the on-site potential energy is bigger than the kinetic energies, that is, in the ``insulator" phase: $U,\mu> t,k$. Due to the very high symmetry of the Hamiltonian in the initial subspace, the system is basically integrable and we can indeed see a short recurrence time. Due to the initial symmetry of the state and the fact that no more than one particle is allowed at every site, the system is very constrained and it is integrable. The entangling power of the Hamiltonian is elevated and despite the fact that the overlap with the initial state is very high, the entanglement is non negligible. The expectation value of every link is the same because of symmetry. For such an initial state, there is no qualitative difference other than different time scales between the ``insulator" and ``superfluid" case.

The time evolution starting from a just less symmetric state is far richer. The state $\ket{\psi_2}= \ket{0000111101}$ is the basis state describing the square with just one diagonal, and again no initial particles. As anticipated, we have studied the model for two sets of parameters. The case (a) is the ``insulator" case with parameters $\mu=U=1; t=k=0.1$. The case (b), or ``superfluid" case has parameters $\mu=U=0.1; t=k=1$.

{\em ``Insulator" case (a).---} In the graph, Fig.~\ref{q2}(a) are plotted the time evolutions of $\langle D_1(t)\rangle , \langle D_2(t)\rangle$ which have initial values of $\langle D_1(0)\rangle= 3,\langle D_2(0)\rangle = 2$. The oscillations of these operators are damped as well and the system is thermalizing towards a state which represents an homogeneous graph. It is very remarkable to see the phenomenon of eigenstate thermalization in such a small system. Recently, there has been a revival in the study of how quantum systems react to a sudden quench in the context of equilibration phenomena in isolated quantum systems, and our results are showing indeed that for such an isolated quantum system, the reduced system can thermalize due to the entanglement dynamics \cite{thermalization, quench}. In Fig.~\ref{q2} (b) are plotted the expectation value of the link operators $P_{12},  P_{24}$ as a function of time In the initial state $\ket{\psi_2}$ we have $\bra{\psi_2}P_{12}\ket{\psi_2}=1, \bra{\psi_2}P_{24}\ket{\psi_2}=0$. The evolution of $\langle P_{12}(t)\rangle$ is almost periodic but we see that on the other hand the oscillations of $\langle P_{24}(t)\rangle$ are damping and the system is thermalizing. The behavior of the $P_{13}$ operator is complementary to $P_{24}$ and at long times  $\lim_{t\rightarrow\infty}\langle P_{24}(t)-P_{13}(t)\rangle = 0$.  In Fig.~\ref{q2} (c) we plotted the entanglement per site measured by the von Neumann entropy $s_j(t)$ as a function of time. The sites considered are again $j=1$ and $j=4$. The two quantities split in two separated bands. Naively, one would expect that the vertices with higher degree are more entangled, but it is not so. Comparing with Fig.\ref{q2} (a) we see that the degrees $\langle D_1(t)\rangle,\langle D_4(t)\rangle$ cross several time and have same time average. Surprisingly, the vertex with consistent higher entanglement is the one that started with higher degree at time zero, {\em when the system was in a separable state}. The system has a memory of the initial state that is revealed in the entanglement dynamics. This means that there are some global conserved quantities that are not detected by any local observable, but are instead encoded in the entanglement entropy, which is a function of the global wavefunction. The thermalization process is also shown in the behavior of the fidelity $\mathcal F (t)$ that presents damped oscillations, see Fig.~\ref{q2} (d). The behavior of the von Neumann entropy $S(t)$ in Fig.~\ref{q2} (d) shows that the reduced system of the particles is indeed evolving as an open quantum system. Though some of the observables are thermalizing, the entanglement dynamics does not show any damping. In Fig.~\ref{q2} (e) is shown the time evolution of the expectation values of the number operators $N_i(t) =S^+_i S^-_i (t)$at the vertices $i=1,4$. The vertices are distinguished by the initial state of the graph, namely by their degree. The average number of particles per site is $\overline{N_i(t)}\sim 0.05$. This means that particles are basically involved in virtual creation/annihilation processes through which the graph acquires a dynamics. The next graphs show that the entanglement dynamics is all but trivial. The following Fig.~\ref{q2} (f) shows the time evolution of the Concurrence $C(t)$ between the vertex $i=2$ and the edge $(2,4)$ which is in the state $\ket{0}$ at the initial time. We see deaths and revivals of entanglement, and the damping of the oscillation signals again some thermalization. 

{\em ``Superfluid" case (b).---} In this case, the kinetic terms $k,t$ are dominant over the potential terms $U,\mu$. The dynamics of this model for these parameters is completely different. We do not have any sign of thermalization. The degree expectation values $\langle D_1(t)\rangle , \langle D_2(t)\rangle$ have a similar oscillating behavior with an even higher overlap, see Fig.~\ref{q3} (a). The link operators $P_{14},P_{24}$, shown in Fig.~\ref{q3} (b) oscillate with no damping and are almost completely overlapped. To such distinct behavior with respect to the insulator case, we find a very strong similarity in the behavior of the entanglement per site $s_i(t)$ (see Fig. \ref{q3} (c) where, again, and in a more pronounced way, there is a splitting in two bands depending on the initial state of the system, and not on the degree (or other interesting observables) of the system during the time evolution. As in the insulator case, the vertex that started off with a higher degree is constantly more entangled than the one that started off with a lower degree, even if in the initial state they are both separable states and during the evolution all the relevant observables overlap strongly and have same time averages. This phenomenon again reveals how the entanglement contains global information on the state of the system that is not revealed in the usual local observables one looks at. The entanglement between edges and particles has a similar behavior than in the ``insulator" case, but it is an order of magnitude greater, which is consistent with the fact that now the terms that couple edges and particles (and thus create entanglement) are larger. The superfluidity is revealed also in the behavior of the fidelity $\mathcal F (t)$ that shows no sign of thermalization with constant amplitude of oscillations, see Fig.~\ref{q3} (d). The average number of particles per site is now $\overline{N_i(t)}\sim 0.52$ and is homogeneous (Fig.~\ref{q3} (e)). Particles are delocalized over the quantum graph with a non vanishing expectation value. The Concurrence $C(t)$ in Fig.~\ref{q3} (f) confirms that there is no thermalization in the system. 

To conclude this section, we have studied the model Eq.(\ref{h}) in the case of hard-core bosons. The usual Hubbard model with hard-core bosons on a fixed graph presents two quantum phases at zero temperature. An insulator phase, when the potential energy of the electrons is dominant, and a superfluid phase, when the kinetic energy is dominant. In our model, the graph interacts with the electrons and the graph degrees of freedom are themselves quantum spins that can be in a superposition. We have studied numerically the entanglement dynamics of the system with four vertices starting from a separable state. The evolution with insulator parameters shows typical signs of thermalization in some of the relevant observables. Moreover, the entanglement dynamics reveals a memory of the initial state that is not captured in the observables. The behavior of the dynamics of the ``superfluid" system is completely different, in the fact that there is no apparent thermalization. The memory effect revealed by the entanglement dynamics is present in an even more pronounced way. There are many open questions to be answered: how the entanglement spectrum behaves and what it reveals of the system, what is the phase diagram of the model at zero temperature in the thermodynamic limit, and a systematic study of the correlation functions in the model. We barely started studying the features of this model that presents formidable difficulties, but that promises to be very rich.

\section{Markov chains analysis of the model}
In this section, we develop a general method to describe the evolution of graphs. We regard Eq.(\ref{h}) as the
Hamiltonian for a classical model and consider a configuration of the system
with a fixed number of edges and particles. The sum of these two quantities
is a constant of the evolution. Moreover, it is safe to assume that almost
all edges can be potentially converted in particles. The reason is simple:
fixing the number of vertices, every connected graph (up to isomorphism) can
be obtained by deleting and adding edges that are part of triangles. With
this approximation, we expect that at long times a dynamical
equilibrium is established between particles and edges. When considering the classical
model, we can {disregard} superpositions and look at the dynamics as
a discrete-time process with characteristics described as follows. For
simplicity, we focus on the complete graph $%
K_{N}=(V(K_{N}),E(K_{N}))$, with set of vertices $V(K_{N})$ and set of edges
$E(K_{N})$. {Note that this constraint is not necessary, since we can start
from any graph containing a triangle. The process simply needs at least one
triangle in order to run. }The process, starting from time $t=0$, can be
interpreted as {a probabilistic} dynamics gradually transforming the
complete graph into its connected {spanning } subgraphs. {These
are subgraphs on the same set of vertices. Methods from the theory of Markov
chains appear to be good candidates to study such a dynamics. We identify
with a \textquotedblleft graph of graphs\textquotedblright\ the phase space
representing all possible states of the system considered. In this way a
random walk on the graph, driven by appropriate probabilities, will allow us to
study the behavior of the Hamiltonian, at least restricting ourselves to
the classical case. Thus, the Hamiltonian transforms graphs into graphs. }%
The Markov chain method {suggests} different levels of analysis: a level
concerning the \emph{support} of the dynamics; a level concerning the \emph{%
distribution of particles}. At the first level, we are interested in
studying graph theoretic properties of the graphs/objects obtained during
the evolution, if we disregard the movement of the particles. {This
consists of studying expected properties of the graphs obtained by running
the dynamics long enough. It is important to remark that the presence of
particles does not modify the phase space, which, if we start from }$K_{N}$%
{, is the set of all connected graphs. The graphs obtained cannot
have more edges than the initial one. Notably, particles only alter the
probability of hopping between elements of the phase space. }At the second
level, we are interested in studying how the particles are going to be
distributed on the vertices of the single graphs, and therefore {in
what measure} the particles determine changes on the graph structure and
{consequently} modify the support of the dynamics. When {%
considering} only the support, the Hamiltonian for the classical model
determines the next process:

\begin{itemize}
\item At time step $t=0$, we delete a random edge of $G_{0}\equiv K_{N}$ and
obtain $G_{1}$.

\item At each time step $t\geq1$, we perform one of the following two
operations on $G_{t}$:

\begin{itemize}
\item \emph{Destroy a triangle:} We delete an edge randomly distributed over
all edges in triangles of $G_{t}$. A \emph{triangle} is a triple of vertices
$\{i,j,k\}$ together with the edges $\{i,j\},\{i,j\},\{j,k\}$.

\item \emph{Create a triangle:} We add an edge randomly distributed over all
pairs $\{i,j\}\notin E(G_{t})$ such that $\{i,k\},\{j,k\}\in E(G_{t})$ for
some vertex $k$.
\end{itemize}
\end{itemize}

This process is equivalent to a random walk on a graph $\mathcal{G}_{N}$
whose vertices are all connected graphs. Each step is determined by the
above conditions. {The Hamiltonian gives a set of rules determining
the hopping probability of the walk. }The theory of random walks on graphs
is a well established area of research with fundamental applications now
ranging in virtually every area of science \cite{lov96}. The main questions to ask
when studying a random walk consist of determining the stationary
distribution of the walk and estimating temporal parameters like the number
of steps required for the walk to reach stationarity. The stationary
distribution at a given vertex is intuitively related to the amount of time
a random walker spends visiting that vertex. In our setting, the walker is a
classical object in a phase space consisting of all connected graphs with
the same number of vertices. Fig. \ref{gg} is a drawing of $\mathcal{G}_{4}$%
, the configuration space of all connected graphs on four vertices. This is
a graph whose vertices are also graphs. The initial position of the walker
is the vertex corresponding to $K_{4}$.

\begin{figure}
  \includegraphics[width=12cm]{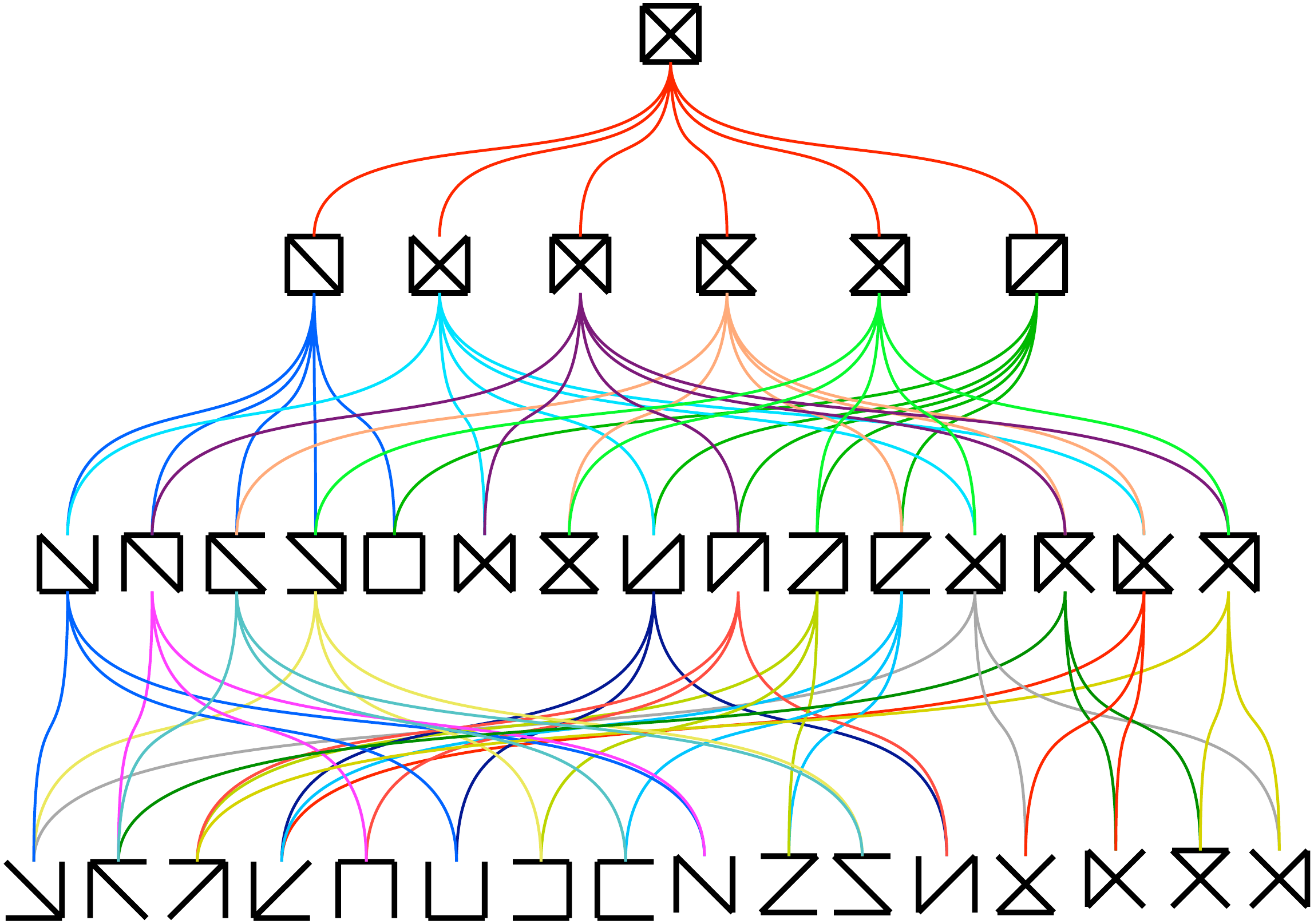}\\
  \caption{The graph $\mathcal{G}_{4}$. The
number of vertices is $38$ and $72$ edges. The number of vertices of $%
\mathcal{G}_{n}$ is exactly the number $d_{n}$ of connected labeled graphs
on $n$ vertices. The number $d_{n}$ satisfies the recurrence $n2^{\binom{n}{2%
}}=\sum\nolimits_{k}kd_{k}2^{\binom{n-k}{2}}$ \protect\cite{wilf94}. The
walker starts from the vertex corresponding to the graph $K_{4}$. Even when
we add particles, the graph $\mathcal{G}_{n}$ remains the support of the
dynamics. The vertices of $\mathcal{G}_{n}$ are the possible states of
classical evolution. In the quantum evolution, we have a weighted
superposition of vertices.}
  \label{gg}
\end{figure}

The graph $%
\mathcal{G}_{N}$ ($N\geq 2$) is connected and bipartite. The number of
vertices of $\mathcal{G}_{N}$ equals the number of connected \emph{labelled}
graphs on $N$ vertices. {We need labels on the vertices to
distinguish between isomorphic graphs. }From the adjacency matrix of a graph
$G$, we can construct the transition matrix $T(G)$ inducing a \emph{simple
random walk} on $G$: $[T(G)]_{i,j}=1/d(i)$ if $\{i,j\}\in E(G)$ and $%
[T(G)]_{i,j}=0$, otherwise. Here, $d(i):=\left\vert \left\{ j:\{i,j\}\in
E(G)\right\} \right\vert $ is the \emph{degree} of a vertex $i$. Notice that
the degrees of the vertices in $\mathcal{G}_{N}$ are not uniform, or, in
other words, $\mathcal{G}_{N}$ is not a \emph{regular} graph. In fact, the
degree of the vertex corresponding to $K_{N}$, which is the number of edges
in this graph, is much higher than the degree of the graphs without
triangles. In Fig. \ref{gg}, it is easy to see that $K_{4}$ has degree $6$
and that the path on $4$ vertices, drawn in the bottom-right corner of the
figure, has only degree $2$.

The evolution of a random walk is determined by applying the transition
matrix to vectors labeled by the vertices encoding a probability
distribution on the graph. The law $\left( T(G)^{T}\right) ^{t}\mathbf{v}%
_{0}^{(i)}=\mathbf{v}_{t}$ gives a distribution on $V(G)$ at time $t$, with
the walk starting from a vertex $i$. The vector $\mathbf{v}_{0}^{(i)}$ is an
element of the standard basis of $\mathbb{R}^{N}$. The vector $\mathbf{v}%
_{t}=(\mathbf{v}_{t}^{(1)},\mathbf{v}_{t}^{(2)},...,\mathbf{v}_{t}^{(N)})^{T}
$ is a probability distribution, being $\mathbf{v}_{t}^{(i)}$ the
probability that the walker hits vertex $i$ at time $t$. The distribution $%
\pi =\left( d(i)/2\left\vert E(G)\right\vert :i\in V(G)\right) $ is the
\emph{stationary distribution}, that is $T(G)\pi \mathbf{=}\pi $. If $G$ is
connected and nonbipartite then $\lim_{t\rightarrow \infty }\left(
T(G)^{T}\right) ^{t}\mathbf{v}_{0}^{(i)}=\pi $ \cite{lov96}. The stationary
distribution is independent of the initial vertex. Therefore, a walk on $%
\mathcal{G}_{N}$ can start from any vertex and the asymptotic dynamics
remains the same. It is simple to see that $\mathcal{G}_{N}$ is bipartite.
Then a random walk does not converge to a stationary distribution, but it
\emph{oscillates} between two distributions with support on the graphs with
an odd and an even number of edges, respectively. In fact, for a bipartite
graph $G$ with $V(G)=A\cup B$, we have the following: $\lim_{t\rightarrow
\infty ,\text{even}}\left( T(G)^{T}\right) ^{t}\mathbf{v}_{0}^{(i)}=\pi
_{even}$ with $[\pi _{even}]_{i}=d(i)/\left\vert E(G)\right\vert $ if $i\in A
$ and $[\pi _{even}]_{i}=0$, otherwise; analogously for $\lim_{t\rightarrow
\infty ,\text{odd}}\left( T(G)^{T}\right) ^{t}\mathbf{v}_{0}^{(i)}=\pi _{odd}
$. For instance, it follows that the stationary distribution of a random
walk starting from any vertex of $\mathcal{G}_{4}$ oscillates between the
two distributions%
\begin{equation*}
\pi _{odd}=(\underset{1}{\underbrace{0}},\underset{6}{\underbrace{%
1/12,...,1/12},}\underset{15}{\underbrace{0,...,0}},\underset{4}{\underbrace{%
1/24,...,1/24}},\underset{12}{\underbrace{1/36,...,1/36}})
\end{equation*}%
and%
\begin{equation*}
\pi _{even}=\underset{1}{(\underbrace{1/12}},\underset{6}{\underbrace{0,...,0%
},}\underset{4}{\underbrace{5/72,...,5/72}},\underset{3}{\underbrace{%
1/31,...,1/31}},\underset{8}{\underbrace{5/72,...,5/72}},\underset{4}{%
\underbrace{0,...,0}},\underset{12}{\underbrace{0,...,0}}).
\end{equation*}%
In particular, for $\pi _{odd}$, $\sum\nolimits_{i:d(i)=6}[\pi _{odd}]_{i}=%
\frac{1}{2}$, $\sum\nolimits_{i:d(i)=4}[\pi _{odd}]_{i}=\frac{1}{6}$, $%
\sum\nolimits_{i:d(i)=2}[\pi _{odd}]_{i}=\frac{1}{3}$; for $\pi _{even}$, $%
\sum\nolimits_{i:d(i)=6}[\pi _{odd}]_{i}=\frac{1}{12}$, $\sum%
\nolimits_{i:d(i)=5}[\pi _{odd}]_{i}=\frac{5}{6}$, $\sum%
\nolimits_{i:d(i)=2}[\pi _{odd}]_{i}=\frac{1}{12}$. Now, what is the most
likely structure of a graph/vertex of $\mathcal{G}_{N}$ in which the random
walker will spend a relatively large amount of time? In other words, where
are we going to find the walker if we wait long enough and what are the
typical characteristics of that graph or set of graphs? From the above
description, one may answer this question by determining the stationary
distribution of the walk in $\mathcal{G}_{N}$. We do not have immediate
access to this information, because we do not know the eigenstructure of $T(%
\mathcal{G}_{N})$. For this reason, we need some way to go around the
problem. We can still obtain properties of the asymptotics by making use of
standard tools of random walks analysis. In particular, as a first step, we
are able to estimate the number of edges in the most likely graph in $%
\mathcal{G}_{N}$. Even if this information is not particularly accurate and
it is far from being sufficient to determine the graphs, it still can give
an idea of their structure. The probability $\pi (G)$ that the walk will be
at a given graph $G$ after a large number of time steps is given, up to a
small error term, by the stationary distribution $\pi $ of the walk. As we
have mentioned above, $\pi (G)$ is given by the number of possible
transitions from $G$ in the walk, divided by a normalizing constant $E$
which is independent of $G$. Based on this, it is possible to find the
expected number of edges in a graph visited by the walk. We will provide
here a sketch of the proof. A more extensive discussion is the Appendix. In
total there are $\binom{\binom{N}{2}}{k}$ graphs with $N$ vertices and $k$
edges. By some of the classical results in the theory of random graphs \cite%
{Bol} we know that for $k\geq \frac{(1+\epsilon )\log N}{N}$, the
probability that a random graph is connected converges to $1$ as $N$ grows.
Let us recall briefly that a random graph is a graph whose edges are chosen
with a fixed probability, equal and independent for each pair of vertices.
The probability that a graph is visited by our walk will have $k$ edges is
given by $\sum_{V(\mathcal{G}_{N}):|E(\mathcal{G}_{N})|=k}\pi (\mathcal{G}%
_{N})$, where the sum is in fact taken over all connected graphs with $k$
edges. In order to estimate this sum we must know how much $\pi (\mathcal{G}%
_{N})$ can vary. However by the results mentioned earlier $\pi (\mathcal{G}%
_{N})$ cannot be larger than $\binom{N}{2}/E$ and not smaller than $(N-1)/E$%
, since that is the largest and smallest number of edges in a connected
graph, and every transition in the walk can be associated with an edge in
the current graph. Hence, the probability that the graph will have $k$ edges
will lie between $\binom{\binom{N}{2}}{k}\binom{N}{2}/E$ and $\binom{\binom{N%
}{2}}{k}(N-1)/E$. However for large $N$ and $k$ this value is completely
dominated by the first term, which is of order $2^{N^{2}}/N$ for $k=\binom{N%
}{2}/2$. A more careful use of these estimates shows that the expected
number of edges will be close to $\binom{N}{2}/2\sim N^{2}/4$. Furthermore
this results will hold true for any walk where the ratio between
probabilities for the most and least likely graphs is not exponentially
large in $N$. This observation tells us that the most likely graphs obtained
during the process tend to have less edges than regular objects as
lattice-like graphs. The number of edges in a square lattice with $n^{2}$
vertices is $2n\left( n+1\right) $. For a cubic lattice on $n^{3}$ vertices
this number is $3n\left( n+1\right) ^{2}$, from the general formula $%
dn\left( n+1\right) ^{d-1}$, where $d$ is the dimension of the lattice.
It seems
natural to try to establish a relation between our walk and random graphs.
After a first analysis, such a relation does not appear obvious. In fact,
the walk on $\mathcal{G}_{N}$ is based upon a \emph{locality principle}
which is not usually defined when considering random graphs. An attempt to
implement this principle for random graphs would consist in constructing Erd\"{o}s-Renyi 
graphs starting from a random tree instead of the empty graph,
that is, the graph with zero edges. A random tree, insures connectivity.
Each pair of vertices at distance two is then joined with a probability $p$.
If we keep adding and deleting edges, we obtain a dynamics similar to the
one induced by our Hamiltonian. It is important to observe that  the
differences with the standard notion of random graph are essentially two:
vertices at distance larger than two cannot be joined with a single step of
the process; there is an additional probability of deleting edges.

Let us
keep in mind that so far we have not consider particles. Indeed, we have
studied only a random walk on $\mathcal{G}_{N}$, where this is the space of
objects obtained by deleting and adding edges that form triangles. However,
our Hamiltonian describes an evolution including particles. Each edge
deletion creates two particles sitting at the end vertices of the deleted
edge. These particles are free to move in the graph. Creation of another
edge will depends on the number of particles. Only when two particles are
located on two different vertices at distance two from each other, then
we have a nonzero probability of creating an edge between such vertices and
therefore creating a new triangle. Including particles, we can define the
following process:

\begin{itemize}
\item At time step $t=0$, we delete a random edge of $G_{0}:=K_{N}$ and
obtain $G_{1}$.

\item At each time step $t\geq1$, we perform one of the following two
operations on $G_{t}$:

\begin{itemize}
\item \emph{Destroy a triangle:} We delete an edge randomly distributed over
all edges in triangles of $G_{t}$. When deleting an edge we \emph{create}
two (indistinguishable) particles. Each particle is located on a vertex of
the graph according to the stationary distribution over $G_{t+1}$. This
reflects the assumption that the particles thermalize.

\item \emph{Create a triangle:} With a certain probability, we add an edge $%
\{i,j\}\notin E(G_{t})$ such that $\{i,k\},\{j,k\}\in E(G_{t})$ for some
vertex $k$. The probability of adding this edge is proportional to the
probability of finding a particle at vertex $i$ and a particle at vertex $j$
at the same time $t$. When adding an edge, we \emph{destroy} two particles.
Specifically, the particles located in the two end vertices.
\end{itemize}
\end{itemize}

Notice that the probability of deleting an edge is independent of the number
of particles in the graph and their locations. On the other side, the
probability of adding an edge is fundamentally connected to the number of
particles. Higher is the number of particles in the graph $G_{t}\in V(%
\mathcal{G}_{N})$ and higher is the probability of adding edges. The process exhibits a conservative 
behavior since the number of particles is always
\begin{equation}
2\binom{N}{2}-\left\vert E(G_{t})\right\vert .  \label{part}
\end{equation}%
So, the dynamics is again equivalent to a random walk on the graph $\mathcal{%
G}_{N}$. This time the random walk is not a simple random walk, since the
probability of each step is determined by the above conditions. In the
transition matrix $T(\mathcal{G}_{N})$ we can have $[T(\mathcal{G}%
_{N})]_{i,j}\neq \lbrack T(\mathcal{G}_{N})]_{i,k}\neq 0$, whenever $j\neq k$%
. This fact gives different nontrivial weights on the edges of $\mathcal{G}%
_{N}$. The transition matrix of the walk is then not necessarily symmetric
and we need a normalization factor to keep it stochastic (\emph{i.e.}, the
sums of the elements in each row is $1$). Whenever an edge is deleted two
particles are created. In the simplest version of our model all particles
are distinguishable and at each time step of the walk all particles are
redistributed according to a random walk on each graph/vertex of $\mathcal{G}%
_{N}$. By Eq. (\ref{part}), there are $N^{2p}$ possible particle
configurations. By the standard behavior of a random walk, particles will
tend to cluster at vertices with high degree. For this model the state of
the walk will consist both of the current graph $G$ and the vector $\mathbf{x%
}$ of positions of all particles. Here the number of possible transitions
will depend both on the structure of $G$, as before, and the number of
particles. Since the number of particle configurations grows rapidly while
the number of edges decreases, the walk will concentrate on connected graphs
with few edges, rather than the denser graphs favored by the model without
particles. If we make a rough estimate of the number of states corresponding
to graphs with $\binom{N}{2}/2$ edges, we see that they are fewer than $N^{%
\binom{N}{2}}2^{\binom{N}{2}}$ and that the number of states corresponding
to graphs with $\mathcal{O}(N)$ edges are more than $N^{(2-\epsilon )\binom{N%
}{2}}$, for any $\epsilon >0$. A comparison argument like the one used for
the case without particles then shows that the expected number of edges for
a graph/vertex will be $o(N^{2})$. It has to be remarked that an \emph{ad
hoc }tuning of the deletion probability for each edges should plausibly
allow to obtain sparser or denser graphs. We have give a rough bound on the
number of edges in a typical graph obtained via the Hamiltonian in Eq.(\ref%
{h}). The bound does not contradict the possibility that such a graph has an
homogeneous structure like a lattice.
Additionally to the number of edges, it may be worth to have some information about cliques.
A \emph{clique} is a complete subgraph. Cliques are then the \emph{densest} regions in a graph.
Knowing the size of the largest cliques gives a bound on the maximum degree and
clearly tells about the possibility of having dense regions.
In the Appendix we will prove that the growth of the largest cliques is logarithmic with respect to the number of vertices.
This behavior also occurs for random graphs.

Numerical simulations were performed to obtain information on the behavior of the classical system under different initial conditions. We are going to discuss the case of  the complete graph $K_N$ as initial state. Complete graphs are interesting for several reasons. First of all, every edge of $K_N$ is eligible for interaction. This implies that edges rapidly transform in particles. As we can see in Fig. \ref{complete_final} (a), the number of particles increases rapidly until it reaches an equilibrium value $\tilde N_0(N)$. The number of steps to reach the equilibrium distribution is the same for all the graph sizes, and is of the order of the inverse of the only \textit{time scale} introduced, given by $\sim P_i^{-1}$. It is interesting to understand the equilibrium distribution of the degree for the various graphs $K_N$. To obtain a better shape for this distribution, we increased the number of simulations from 30 to 60. The result can be seen in Fig. \ref{complete_final} (b). We find that the distribution is Poisson ($D_e$ is the degree), as it is for random graphs:
$$P_N(D_e)=\frac{1}{R}e^{- \frac{(D_e-f(N))^2}{Q(N)}} $$
where $R$ is a normalization constant. In Fig. \ref{complete_final} (a) we find that the function $f(N)$ is, for the graph $K_N$, given by
$$f(N)=\frac{N}{2}, $$
while the function $Q(N)$ of Fig. \ref{complete_final} (b)  is
$$Q(N)=\frac{\sqrt{N}}{2}. $$
It has to be remarked that this result agrees for large $N$ with the combinatorial proof in the last section.

\begin{figure}
  \includegraphics[width=18cm]{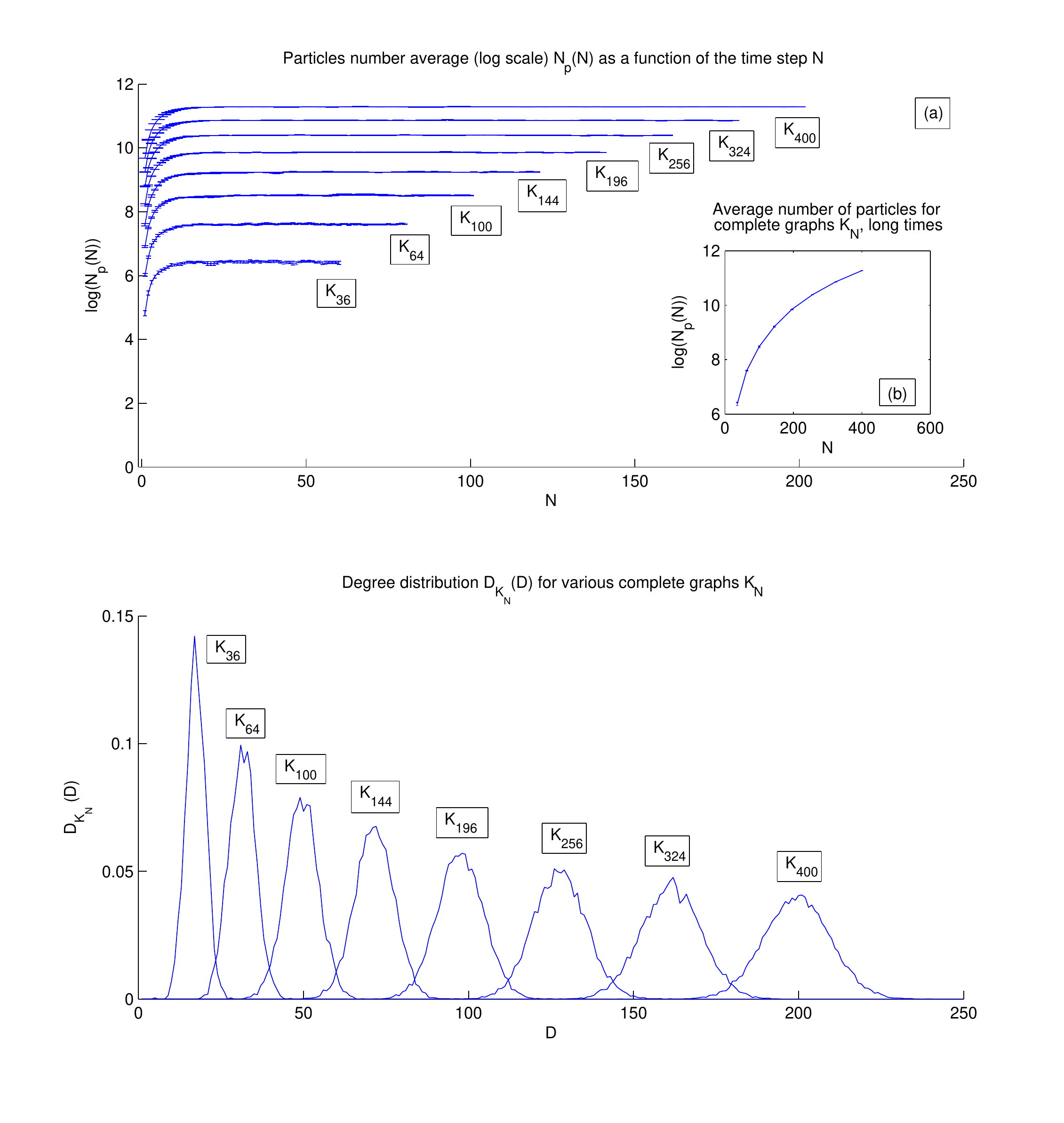}
  \label{complete_final}
  \caption{ \textbf{a)} Plot of the logarithm of the average number of particles $N_p$ as a function of the number of steps. The initial conditions were complete graphs with $N$ vertices, $K_N$, with hopping probability $P_h=1$ and interacting probability $P_i=0.1$. \textbf{b)} In this plot we have the same quantity of (a), $\log(N_p(N))$, plotted as a function of the graph size $N$ at long times (when at equilibrium). \textbf{c)} Plot of the  equilibrium degree distributions for each graph $K_N$. The hopping and interaction probabilities for each simulation are $P_h=1$ and $P_i=0.1$ respectively. The distributions are obtained averaging over 60 simulations. At equilibrium, we obtain Poisson distributions centered on $\frac{N}{2}$, as shown in Fig. (\ref{complete_final2}) (b), and variance $Q(N)=\frac{\sqrt{N}}{2}$, as shown in Fig. (\ref{complete_final2}) (a). This agrees with standard results of the theory of random graphs.} 
\end{figure}
\begin{figure}
  \includegraphics[width=18cm]{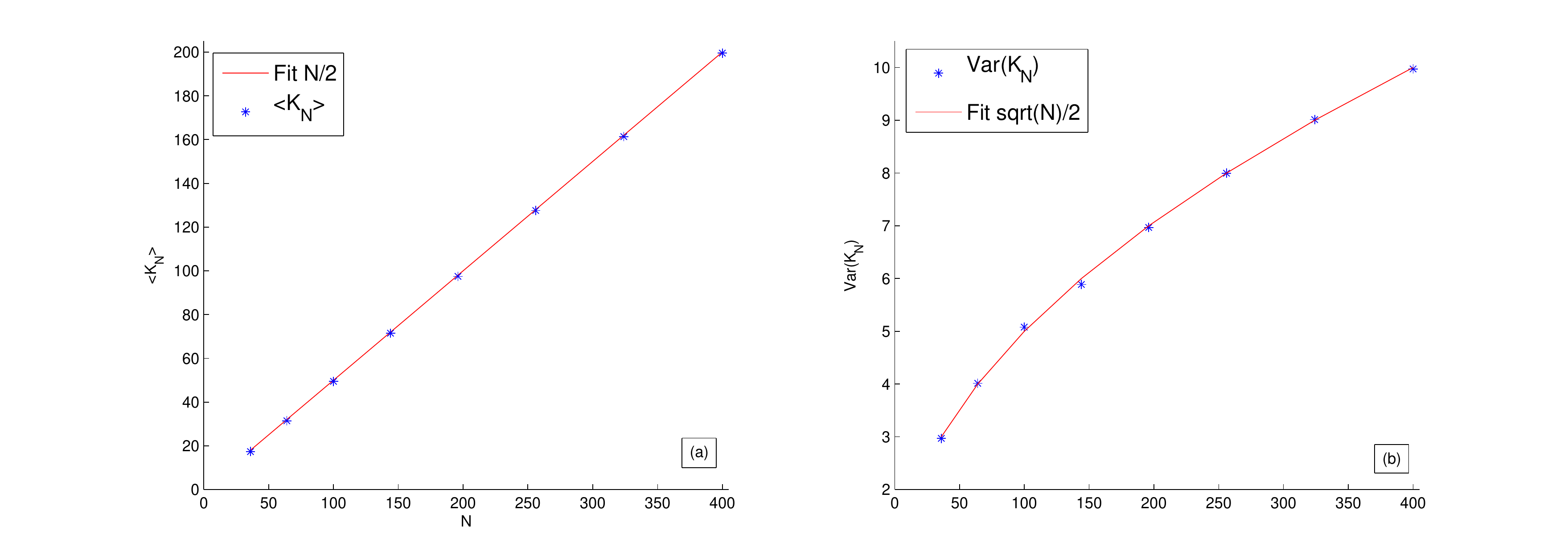}
  \label{complete_final2}
  \caption{\textbf{a)} Plot of the average values of the equilibrium degree distributions of Fig. (\ref{complete_final}) (c). The red line is the fit, given $f(N)=\frac{N}{2}$.  This agrees perfectly  with the various $K_N$. \textbf{b)} Plot of the variance of the equilibrium degree distributions of Fig. (\ref{complete_final}) (c). The red line is the fit, giving $Q(N)=\frac{\sqrt{N}}{2}$. }
\end{figure}


\section{Conclusions}

The work presented here was motivated by the possibility that the answer to the problem of quantum gravity lies in the direction of emergent gravity, and, in particular, in a condensed matter perspective to emergent gravity.
Along these lines, in previous work \cite{KonMarSev,HamMarPreSev},  {\em quantum graphity} was proposed and analyzed, a background independent spin system for emergent locality, geometry and matter.  The quantum states of this system are dynamical graphs whose connectivity represents locality.  The Hamiltonian of \cite{KonMarSev} is not unitary: the universe starts at a high energy configuration (non-local) and evolves to a low energy one (local).  This has clear limitations when applied to a cosmological context.  The present model was originally intended as an energy-conserving version of \cite{KonMarSev}, in which graph edges can be deleted, matter created and vice versa.  The Hamiltonian we used is essentially an extension of the Hubbard model to a dynamical lattice.

In our system, the basic building blocks of the theory are not events, but quantum physical systems $S_i$, represented by a finite-dimensional Hilbert space $\mathcal H_i$ and a Hamiltonian $H_i$.  That is, instead of an event, we have the space of all possible states of $S_i$ and the dynamical rules for the time evolution of these.  The aim is to study the relationships among the $S_i$ and find geometry as the emerging structure imposed by these relationships, independently of their state.  In fact, what we have can be considered as a unification of matter and geometry to matter only.  Certain configurations (bound states) of matter play the role of spatial adjacency in the dynamical sense that they will determine whether interaction of other particles is allowed or not.  In that sense, geometry is not fundamental but rather a convenient way to simplify the fundamental evolution of matter. From the condensed matter point of view, we have presented a Hubbard model in which the particles hop on a graph whose shape is determined by the motion of the particles and it is itself a quantum variable.

We simulated the quantum system for a complete graph with $4$ vertices and hard core bosons and we analyzed the entanglement dynamics of the system, including the one between the particle and edge degrees of freedom. We argued that for the weakly interacting system, entanglement and loss of unitarity for the reduced system can be seen in presence of very high curvature. Moreover, the eigenstate thermalization of the model is studied under two different sets of parameters. Thermalization occurs when potential energy dominates over kinetic energy.

The model discussed in this paper has some features of other models
studied in a diversity of contexts and with different purposes. The first
central aspect is a quantum dynamics involving a set of graphs. Quantum
evolution on graphs is a subject of study related to spin systems, as a
generalization of spin chains (see the review by Bose \cite{sou}), and
quantum walks on graphs, where a particle undergoes a Schr\"{o}dinger
dynamics hopping between the vertices (see the review by Kempe \cite{kem})
where, at each time step, a particle is in a
superposition of different vertices. Of course, the main difference with respect to our model is in the fact that in our model the graph is evolving in time as part of the system's evolution. Focusing on another area, it may be
interesting to highlight a parallel with the work of Gudder \cite{gud}, who
studied discrete-space time building on ideas of Bohm \cite{bohm}. In the
model described by Gudder, the graph is interpreted as a discrete
phase-space in which the vertices represent discrete positions which a
particle can occupy, and the edges represent discrete directions that a
particle can propagate. This setting was primarily introduced to describe
the internal dynamics of elementary particles. From this perspective, each
particle is associated to a graph: vertices represent quarklike constituents
of a particle and edges represent interaction paths for gluons which are
emitted and absorbed by the vertices. Another aspect of our model is that we have
 many particles evolving at the same time. In the mathematical
literature there is a growing number of examples of random walks with
multiple particles/agents. A recent paper by Cooper \emph{et al}.\ \cite{coop}
studied properties of multiple random walks on a (fixed) graph assuming that the
interaction between particles gives rise to various phenomena, like
particles sticking together, annihilating each other, \emph{etc. }It is
important to mention that there are differences between walks with a single
agent and walks with many agents. For instance, there are various settings
in which $k$ random walks visit all the nodes of a graph in expectation $%
\Omega (k)$-times faster than the case of a single walker \cite{al}. Works
addressing random walks on evolving graphs have been considered only
recently, motivated by robotic exploration of the Web \cite{avi, coop1}.

In a unitary system for cosmological evolution, the interesting question to ask is whether the system has long-lived metastable states.  We studied this question rigorously, using analysis in terms of Markov chains, and also gave an intuitive argument, in both cases finding that dynamical equilibrium is reached when, starting from zero initial particles, half the initial number of edges are destroyed.
In the model, there is a candidate configuration for a toy mechanism for attraction.  A toy trapped surface can be modeled as the boundary of a region containing a highly connected subgraph.  There is no singularity.  To probe this region one can use Wen's hamiltonian for light \cite{Wen}, with a coupling that is just a perturbation of the initial hamiltonian.  We found that, in the limit of a large number of vertices inside the highly connected region, only a ray perfectly normal to the surface can escape.

An important issue to emphasize before closing this paper is that our model assumes the existence of a notion of time and of time evolution as given by a Hamiltonian, as opposed to the constrained evolution of canonical pure gravity.  It is a general question for all  
 condensed matter approaches to quantum gravity whether such evolution is consistent with the diffeomorphism invariance of general relativity.  While it is not possible to settle this question without first knowing whether the condensed matter microscopic system has a low energy phase which is general relativity, we can make a few comments, as well as point the reader to more extensive discussion of this issue elsewhere \cite{me,FQXi}.   In general, 
there are two possible notions of time:  
 the time related to the $g_{00}$ component of the metric describing the geometry at low energy and the time parameter in the fundamental microscopic Hamiltonian.  
Let us call the first {\em geometric time} and the second {\em fundamental time}.    
In our geometrogenesis context, it is clear that the geometric time will only appear at low energy, when geometry appears.   The problem of the emergence of geometric time is  the same as the problem of the emergence of space, of geometry.  The constrained evolution of general relativity, often called 
``time does not exist'', refers to geometric time.  By making the geometry not fundamental, we are able to make a distinction between the geometric and the fundamental time, which opens up the possibility that, while the geometric time is a symmetry, the fundamental time is real.  It is important to note that the
relation between geometric and fundamental time is non-trivial and that the existence of a fundamental time does not necessarily imply a preferred geometric time.  We also note that, in the presence of matter in general relativity, a proper time can be identified.  The particular system studied here has matter and in that sense it is perhaps more natural that it also has a straightforward notion of time.

The toy model we presented here is very basic and there are several features we do not expect to see yet.  For instance, there is nothing in the Hamiltonian to encourage the system to settle in metastable states that are regular graphs.  Emergence of geometric symmetries such as Friedman-Robertson-Walker symmetries was not a goal at this stage but can be incorporated in future work by additional terms in the hamiltonian as in \cite{KonMarSev}, or possibly also by introducing causality restrictions as in \cite{CDT}.  Finally, we believe that this model is interesting from the condensed matter point of view. Condensed matter systems are always defined on a given lattice. In this model, the lattice itself is a quantum variable. Such a model, like the spin model of Eq.\ (\ref{onehalf}) can be realized experimentally in a system of quantum dots. We believe this is a novel way to think of condensed matter systems and, in perspective, potentially fruitful for the study of novel quantum phases.

There are many potential generalizations and extensions of the model. Such extensions concern primarily the phase space and the nature of the particles. One possible generalization is to allow each edge to introduce a new vertex as a midpoint, possibly by absorbing
particles. This gives homeomorphic graphs but with a growing phase space. The models could also be modified by allowing the degrees of the vertices
          to influence how the edges and the particles interact. A consequence would be a dynamical graph in which different regions are modified with different speeds.

\section{Acknowledgements}

Research at Perimeter Institute for Theoretical
Physics is supported in part by the Government of Canada through NSERC and
by the Province of Ontario through MRI. This project was partially supported
by NSERC, a grant from the Foundational Questions Institute (fqxi.org), and a grant
from xQIT at MIT. SS is a Newton International Fellow. We are grateful to the Keck Center for Quantum Computing for hospitality at MIT where this work was initiated.

This work was made possible by the facilities of the Shared Hierarchical
 Academic Research Computing Network (SHARCNET:www.sharcnet.ca).

\section{Appendix}

The following general result is useful to study the random walks described
in the paper:

\begin{theorem}
\label{ebound}Let $G$ be a graph drawn from \emph{any} random walk on $%
\mathcal{G}_{N}$ with stationary distribution $\pi $. Assume that there
exists $\alpha <2$ such that
\begin{equation}
g(\pi )=\max_{G,H\in V(\mathcal{G}_{N})}\frac{[\pi ]_{G}}{[\pi ]_{H}}%
<e^{N^{\alpha }}.
\end{equation}%
Then

\begin{enumerate}
\item $E\left( |E(G)|\right) =\frac{1}{2}\binom{N}{2}+o(1)$;

\item For $t\geq \sqrt{N^{\alpha }\binom{N}{2}}$,%
\begin{equation*}
\text{Prob}\left[ \left\vert |E(G)|-\frac{1}{2}\binom{N}{2}\right\vert >t%
\right] \leq 8\exp (-N^{\alpha }).
\end{equation*}
\end{enumerate}
\end{theorem}

\begin{proof}
We will prove the theorem by proving a deviation bound for Prob$\left[
\left\vert |E(G)|-\frac{1}{2}\binom{N}{2}\right\vert >t\right] $ in terms of
$t$.

Let
\begin{equation*}
A_{t}=\left\{ G:\left\vert |E(G)|-\frac{1}{2}\binom{N}{2}\right\vert \leq
t\right\} .
\end{equation*}%
Note that $A_{t}$ contains both connected and disconnected graphs.

Let $C$ be a function from the set of all graphs on $N$ vertices to the set $%
\{0,1\}$. Specifically, let%
\begin{equation*}
C(G):=\left\{
\begin{tabular}{ll}
$1$ & if $G$ is connected$;$ \\
$0$ & otherwise$.$%
\end{tabular}%
\ \ \right.
\end{equation*}%
Let $C(N,t)$ denote the probability that a graph drawn uniformly at random
from $A_{t}$ is connected. The probability that a graph drawn from the
random walk belongs to $A_{t}$ is $\sum_{G\in A_{t}}C(G)[\pi ]_{G}$;
likewise for $G\notin A_{t}$. We will now bound the quotient between these
two probabilities. It is well known, see \emph{e.g.} \cite{Bol}, that
\begin{equation*}
C(N,t)=1+o(1)
\end{equation*}%
for, \emph{e.g.}, $t<N/2$. Then
\begin{align}
\frac{\sum_{H\notin A_{t}}C(H)[\pi ]_{H}}{\sum_{G\in A_{t}}C(G)[\pi ]_{G}}&
\leq g(\pi )\frac{\sum_{H\notin A_{t}}C(H)}{\sum_{G\in A_{t}}C(G)}
\label{bound} \\
& \leq g(\pi )\frac{|\overline{A_{t}}|}{|A_{t}|C(N,t)}  \notag \\
& \leq g(\pi )\frac{|\overline{A_{t}}|}{|A_{t}|(1+o(1))}  \notag \\
& \leq 2g(\pi )\frac{|\overline{A_{t}}|}{|A_{t}|}  \notag
\end{align}%
This can be written as
\begin{equation}
2g(\pi )\frac{P_{N,t}}{1-P_{N,t}},  \label{bound1}
\end{equation}%
where
\begin{equation*}
P_{N,t}=|\overline{A_{t}}|/2^{\binom{N}{2}}
\end{equation*}%
is the probability that a binomial random variable with distribution $%
\mathrm{Bin}({\binom{N}{2}},\frac{1}{2})$ deviates more than $t$ from its
expectation. Using the Chernoff bound we have then
\begin{equation*}
P_{N,t}\leq 2\exp \left( \frac{-2t^{2}}{\binom{N}{2}}\right)
\end{equation*}%
If $t=\sqrt{N^{\alpha }\binom{N}{2}}$ we can thus bound (\ref{bound1}) as
\begin{align}
2g(\pi )\frac{P_{N,t}}{1-P_{N,t}}& \leq 2g(\pi )4\exp (\frac{-2t^{2}}{\binom{%
N}{2}})  \label{bound2} \\
& =8g(\pi )\exp (-2n^{\alpha })  \notag \\
& =8\exp (-N^{\alpha })  \notag
\end{align}%
Since the denominator in the first step of Eq. (\ref{bound}) is less than $1$
we have the following bound for our walk
\begin{equation}
\text{Prob}\left[ H\notin A_{t}\right] \leq \sum_{H\notin A_{t}}C(H)[\pi
]_{H}\leq 8\exp (-N^{\alpha })
\end{equation}%
For our range of $\alpha $ the value of $t$ is $o\left( \binom{N}{2}\right) $
which means that for a graph from $A_{t}$ the number of edges is $\frac{1}{2}%
\binom{N}{2}+o(1)$ and the contribution to the expected number of edges from
graphs not in $A_{t}$ is between zero and $\binom{N}{2}8\exp (-N^{\alpha })$%
. This is $o(1)$. Thus the total expectation is $\frac{1}{2}\binom{N}{2}%
+o(1) $.
\end{proof}

\begin{corollary}
\label{tbound}Let $G$ be a graph drawn from a simple random walk on $%
\mathcal{G}_{N}$. Then%
\begin{equation*}
E(|E(G)|)=\frac{1}{2}\binom{N}{2}+o(1)
\end{equation*}
\end{corollary}

\begin{proof}
We know that the underlying graph of this random walk is bipartite and that
in the asymptotic limit the stationary distribution oscillates between $\pi
_{even}$ and $\pi _{odd}$ depending on whether we have taken an even or an
odd number of steps. However if we start a new walk by not changing the
graph in the first time step with probability $1/2$, the new stationary
distribution will be
\begin{equation}
\pi =\frac{1}{2}(\pi _{even}+\pi _{odd}).  \label{bip}
\end{equation}%
Since we are looking at graphs with $N$ vertices, $d(i)$ is at most $\binom{N%
}{2}$ and not less than $N-2$. Hence,
\begin{equation*}
g(\pi )\leq \frac{N^{2}-N}{2n-4}.
\end{equation*}%
and the corollary follows from Theorem \ref{ebound}.
\end{proof}

\bigskip

Note that here we only used the fact the walk
converges towards a stationary distribution on $\mathcal{G}_{N}$ and that $%
g(\pi )$ is bounded by a polynomial in $N$ for this walk. The same result
holds for any form of walk on $\mathcal{G}_{N}$ for which $g(\pi )$ is not
exponential in $N$.

Let us recall that a random event happens \emph{asymptotically almost surely}%
, or \emph{a.a.s}, if the probability for the event is $1-o(1)$. A \emph{%
clique} in a graph is a subgraph isomorphic to the complete graph. The \emph{%
clique number} of a graph $G$, denoted by $\omega(G)$, is the number of
vertices of the largest clique in $G$.

\begin{theorem}
Let $G$ be a graph drawn from a simple random walk on $\mathcal{G}_{N}$.
Then there exist constants $c_{1}<c_{2}$ such that a.a.s the clique number $%
\omega (G)$ satisfies
\begin{equation*}
c_{0}\log (N)-c_{1}\leq \omega (G)\leq c_{0}\log (N)+c_{2}.
\end{equation*}
\end{theorem}

\begin{proof}
It is well known, see, \emph{e.g.}, Chapter 11 of \cite{Bol}, that the
expected clique number of a uniform random graph with edge probability $%
\frac{1}{2}$ is
\begin{equation*}
\frac{2}{\log 2}\log (N)=c_{0}\log N
\end{equation*}%
and that the following concentration bounds hold
\begin{equation}
\text{Prob}\left[ \omega (G)-c_{0}\log N\geq r\right] <N^{-r},  \label{ou}
\end{equation}%
\begin{equation}
\text{Prob}\left[ c_{0}\log N-\omega (G)\geq r\right] <N^{-\lfloor 2^{\frac{%
r-2}{2}}\rfloor }  \label{ol}
\end{equation}%
We can now proceed in the same way as in the proof of Theorem \ref{ebound}
using the set
\begin{equation*}
B_{t}^{u}=\{G|\omega (G)-c_{0}\log N\}|\leq t\}
\end{equation*}%
to give a bound on the upper tail probability and
\begin{equation*}
B_{t}^{l}=\{G|c_{0}\log N-\omega (G)\}|\leq t\}
\end{equation*}%
for the lower tail probability, together with the bound on $g(\pi )$ from
the proof of Corollary \ref{tbound} and the concentration bounds from the
inequalities (\ref{ou}) and (\ref{ol}). This gives the following
inequalities:

\begin{equation}
\text{Prob}\left[ \omega (G)-c_{0}\log N\geq r\right] <c_{1}\frac{N(N-1)}{N-2%
}N^{-r}
\end{equation}

\begin{equation}
\text{Prob}\left[ c_{0}\log N-\omega (G)\geq r\right] <\frac{N(N-1)}{N-2}%
N^{-\left\lfloor 2^{\frac{t-2}{2}}\right\rfloor }
\end{equation}%
We can now use these inequalities to bound the contributions to the expected
clique number. A clear but lengthly calculation can show that the
contributions from the two tails are asymptotically bounded by two
constants, giving us the bound stated in the theorem.
\end{proof}

\bigskip

A rigorous analysis of the expected number of edges when we consider
particles will be more difficult, since the model corresponds to a random
walk on a directed graph, \emph{i.e.}, a graph in which edges have a
direction. This is associated to an adjacency matrix which is not necessarily
symmetric. There are transitions where, \emph{e.g.}, an edge is deleted and
the particle distribution changes so that the endpoint of the edge do not
have any particles on them, thus making the re-addition of the edge
impossible in the next step. However, for states with a large number of
particles, such as any graph with less than $\frac{1}{2}{\binom{n}{2}}$
edges, the vast majority of transitions will be reversible, as it is in the
case without particles.

Let us consider the number of possible transitions from a state $(G,\mathbf{x%
})$ defined as follows: here $G$ has $n$ vertices, $t$ edges which are part
of triangles, $s$ edges with endpoints at distance $2$, and $p$
indistinguishable particles. If we assume that there are particles at all
vertices we have three types of transitions which can be explicitly
enumerated:

\begin{enumerate}
\item There are ${\binom{n+p-1}{p}}$ transitions which correspond to
redistributing the particles without changing $G$.

\item There are $t{\binom{n+p }{p+1}}$ transitions which correspond  to
deleting an edge and redistributing the $p+1$ particles.

\item There are $s{\binom{n+p-2 }{p-1}}$ transitions which correspond  to
adding an edge and redistributing the $p-1$ particles.
\end{enumerate}

Only the number of transitions of the last type is affected by the
assumption that there are particles at all vertices.

If $p\approx c\frac{1}{2}{\binom{n}{2}}$ we can estimate ${\binom{n+p-1}{p}}$
as%
\begin{eqnarray*}
{\binom{n+p-1}{p}} &\approx &\frac{p^{n-1}}{n!}\left( 1+\frac{n(n-1)}{2p}+%
\mathcal{O}(p^{-2})\right)  \\
&\approx &\frac{p^{n-1}}{(n/e)^{n}}(1+\frac{1}{c}+\mathcal{O}(p^{-2})) \\
&\approx &\frac{e^{n}}{n}\left( \frac{p}{n}\right) ^{n-1}(1+\frac{1}{c}+%
\mathcal{O}(p^{-2})) \\
&\approx &\frac{e^{n}}{n}\left( c\frac{n-1}{2n}\right) ^{n-1}(1+\frac{1}{c}+%
\mathcal{O}(p^{-2})) \\
&\approx &\mathcal{O}\left( \frac{(ce)^{n}}{n}\right)
\end{eqnarray*}

Inserting this into the numbers of transitions given above, shows that for
this model the maximum degree of the transition graph is bounded from above
by a simple exponential, and the minimum degree is of course still greater
than a multiple of $n$. According to Theorem 1 this is not sufficient to
change the expected number of edges from being $\frac{1}{2}{\binom{n}{2}}$.
In order to make this analysis fully rigorous it is also necessary to show
that the states with unoccupied vertices do not make a significant
contribution, which will be lengthy but mostly a technical issue.

The discussion for a case with distinguishable particles will be very
similar but ${\binom{n+p-1}{p}}$ will be replaced by $n^{p}$ which is large
enough to escape Theorem 1.
These counts also give an easy way to implement a simulation  algorithm for
both models. Just pick uniformly among all the  possible transitions from
the current state.
If we do not consider particles or consider indistinguishable
particles, the degrees of the vertices in the graphs will be close to those
of random graphs with probability 1/2. This is more or less for the same reason that Theorem
1 works. The number of \textquotedblleft typical\textquotedblright\ graphs
in $G(n,\frac{1}{2})$ is so large that their behavior will still control
these models.




\end{document}